\documentclass[letterpaper,11pt]{article}

\usepackage{fullpage}

\usepackage{amsmath,amssymb,amscd,amsthm,times,enumerate}
\usepackage{color}
\usepackage[usenames,dvipsnames]{xcolor}
\usepackage[english]{babel}
\usepackage[T1]{fontenc}
\usepackage[latin1]{inputenc}
\usepackage[obeyspaces]{url} 
\usepackage[babel=true]{csquotes}
\usepackage[font={footnotesize,sf}]{caption}
\usepackage{tikz-cd}
\usetikzlibrary{decorations.pathmorphing} 
\usepackage[linktoc=page,colorlinks=true,urlcolor=Mulberry,citecolor=RoyalBlue,linkcolor=Maroon]{hyperref}
\usepackage{soulutf8,wasysym}

\newcommand{\Remark}[1]     {{\small {\sf $\clubsuit \clubsuit
      \clubsuit$ #1 $\clubsuit \clubsuit \clubsuit$}}}

\newcommand{\denselist}{\itemsep 0pt\parsep=1pt\partopsep 0pt}

\newcommand{\R}{\mathbb{R}}
\newcommand{\N}{\mathbb{N}}

\newcommand{\Sspace}{\mathbb{S}}
\newcommand{\Tspace}{\mathbb{T}}

\newcommand{\Bset}{\mathcal{B}}
\newcommand{\Cset}{\mathcal{C}}
\newcommand{\Dset}{\mathcal{D}}

\newcommand{\RCech}[2]{\operatorname{Cech}_{#1}(#2)}
\newcommand{\Cech}[1]{\operatorname{Cech}(#1)}
\newcommand{\Rips}[1]{\operatorname{Rips}(#1)}

\newcommand{\card}{\operatorname{card}}

\newcommand{\Ignore}[1]{ }

\newcommand{\Nerve}[1]{\operatorname{Nrv} #1}

\newcommand{\Star}[2]{\operatorname{St}_{#2}(#1)}

\newcommand{\Offset}[2]{#1^{\oplus #2}}
\newcommand{\MA}[1]{\mathcal{M}_{#1}}

\newcommand{\Split}{\operatorname{Split}}

\newcommand{\Reach}[1]{\operatorname{Reach} \left(#1\right)}
\newcommand{\Radius}[1]{\operatorname{Radius} \left(#1\right)}

\newcommand{\Ball}[2]{\operatorname{B} (#1, #2)}
\newcommand{\OpenBall}[2]{\operatorname{B}^\circ(#1, #2)}

\newcommand{\CapBalls}[2]{\mathcal{B}(#1,#2)}

\newcommand {\scalprod}[2] {{\langle #1 , #2 \rangle}}

\newcommand{\Vertexset}[1]    {\operatorname{Vert}(#1)}

\newcommand{\Hull}[2]{\operatorname{Hull}_{#2}(#1)}
\newcommand{\Conv}[1]{\operatorname{Conv}(#1)}
\newcommand{\Clenchers}[2]{\operatorname{Clenchers}_{#2}(#1)}

\newcommand{\capH}{$\bigcap_{p \in \sigma_0} H(p)$}
\newcommand{\Hfig}{$\Hull{\{x_1,x_2\}}{\alpha'}$}
\newcommand{\nfig}{$\frac{1}{n}$}
\newcommand{\Convfig}{$A \cap \left[\Conv{C_v}\right]^{\oplus \eta_0\rho}$}

\newtheorem{theorem}{Theorem}
\newtheorem{lemma}{Lemma}

\newtheorem{definition}{Definition}

\newcommand{\bb}[2]
           {\mbox{
               \begin{minipage}[l]{20mm}
                 \scriptsize #1\\#2
               \end{minipage}
           }}

\title{Geometry driven collapses for converting a \v
  Cech complex into a triangulation of a nicely triangulable shape\thanks{
Research partially supported by the French \enquote{Agence nationale pour la
Recherche} under grant ANR-13-BS01-0008 TopData.}}

\author{
Dominique Attali\footnote{Gipsa-lab -- CNRS UMR 5216, Grenoble,
  France.
\texttt{Dominique.Attali@gipsa-lab.grenoble-inp.fr}}
\and
André Lieutier\footnote{Dassault systèmes, Aix-en-Provence, France.
\texttt{andre.lieutier@3ds.com}}
}

\begin{document} 
\maketitle

\begin{abstract}
Given a set of points that sample a shape, the Rips complex of the
points is often used to provide an approximation of the shape
easily-computed. It has been proved that the Rips complex captures the
homotopy type of the shape assuming the vertices of the complex meet
some mild sampling conditions. Unfortunately, the Rips complex is
generally high-dimensional. To remedy this problem, it is tempting to
simplify it through a sequence of collapses. Ideally, we would like to
end up with a triangulation of the shape. Experiments suggest that, as
we simplify the complex by iteratively collapsing faces, it should
indeed be possible to avoid entering a dead end such as the famous
Bing's house with two rooms. This paper provides a theoretical
justification for this empirical observation.

We demonstrate that the Rips complex of a point-cloud (for a
well-chosen scale parameter) can always be turned into a simplicial
complex homeomorphic to the shape by a sequence of collapses, assuming
the shape is nicely triangulable and well-sampled (two concepts we
will explain in the paper). To establish our result, we rely on a
recent work which gives conditions under which the Rips complex can be
converted into a \v Cech complex by a sequence of collapses. We
proceed in two phases. Starting from the \v Cech complex, we first
produce a sequence of collapses that arrives to the \v Cech complex,
restricted by the shape. We then apply a sequence of collapses that
transforms the result into the nerve of some covering of the
shape. Along the way, we establish results which are of
  independent interest. First, we show that the reach of a shape can
  not decrease when intersected with a (possibly infinite) collection
  of balls, assuming the balls are small enough. Under the same
  hypotheses, we show that the restriction of a shape with respect to
  an intersection of balls is either empty or contractible. We also
  provide conditions under which the nerve of a family of compact sets
  undergoes collapses as the compact sets evolve over time. We believe
  conditions are general enough to be useful in other contexts as
  well.
\end{abstract}


\clearpage
\section{Introduction}


This paper studies the problem of converting a \v Cech complex whose
vertices sample a shape into a triangulation of that shape using
collapses. Even if the present paper focuses exclusively on the \v
Cech Complex, it has also implications on the simplification of Rips
complexes by sequences of collapses, due to a recent result
in~\cite{Rips-cgta-AttaliLieutierSalinas}.

Imagine we are given a set of points that sample a shape and we want
to build an approximation of the shape from the sample points. An
often used approach consists in outputting the Vietoris-Rips complex
of the points (see for instance
\cite{carlsson2006algebraic,silva07:_cover,tahbaz2010distributed}). Formally,
the {\em Vietoris-Rips complex} of a set of points $P$ at scale
$\alpha$ is the abstract simplicial complex whose simplices are
subsets of points in $P$ with diameter at most $2\alpha$. For brevity,
we shall refer to it as the Rips complex. The Rips complex is an
example of a flag complex --- the maximal simplicial complex with a
given 1-skeleton.  As such, it enjoys the property to be completely
determined by its 1-skeleton which therefore offers a compact form of
storage easy to compute. Moreover, the Rips complex is able to
  reproduce the homotopy type of the shape in certain situations
  \cite{hausmann1995vietoris,latschev2001vietoris,chambers2010vietoris,socg10-convex}.
  Precisely, Hausmann proved in \cite{hausmann1995vietoris} that if
  the shape $A$ is a compact Riemannian manifold, then the Rips
  complex with vertex set $A$ is homotopy equivalent to $A$ when the
  scale used to build the Rips complex is small
  enough. In~\cite{latschev2001vietoris}, Latschev extended this
  result to Rips complexes with vertex set a metric space (possibly
  finite) whose Gromov-Hausdorff distance to the shape is
  small. In~\cite{Rips-cgta-AttaliLieutierSalinas}, a variant has been
  established in a different framework: shapes are assumed to be
  subsets of $\R^d$ with a positive $\mu$-reach and Rips complexes are
  built on finite samples of the shapes using the Euclidean
  distance. The latter result makes the Rips complex an appealing
  object for reconstructing shapes living in high dimensional spaces,
  as for instance in machine learning.

Unfortunately, the dimension of the Rips complex can be very large,
compare to the dimension of the underlying shape it is suppose to
approximate. This suggests a two-phase algorithm for shape
reconstruction. The first phase builds the Rips complex of the data
points, thus producing an object with the right homotopy type. The
second phase simplifies the Rips complex through a sequence of {\em
  collapses}.  Ideally, after simplifying the Rips complex by
repeatedly applying collapses, we would like to end up with a
simplicial complex homeomorphic to the underlying shape.

Yet it is not at all obvious that the Rips complex whose vertices
sample a shape contains a subcomplex homeomorphic to that shape. Even
if such a subcomplex exists, is there a sequence of collapses that
leads to it? Certainly if we want to say anything at all, the geometry
of the complex will have to play a key role. As evidence for this,
consider a simplicial complex whose vertex set is a noisy point-cloud
that samples a 0-dimensional manifold and suppose the complex is
composed of a union of Bing's houses with two rooms, one for each
connected component in the manifold. Each Bing's house is a
2-dimensional simplicial complex which is contractible but not
collapsible. Thus, the complex carries the homotopy type of the
0-dimensional manifold but is not collapsible. Fortunately, it seems
that such bad things do not happen in practice, when we start with the
Rips complex of a set of points that samples \enquote{sufficiently
  well} a \enquote{nice enough} space in $\R^d$. The primary aim of
the present work is to understand why. For this, we will focus on the
{\em \v Cech complex}, a closely related construction. Formally, the
\v Cech complex of a point set $P$ at scale $\alpha$ consists of all
simplices spanned by points in $P$ that fit in a ball of radius
$\alpha$. In \cite{Rips-cgta-AttaliLieutierSalinas}, it was proved
that the Rips complex can be reduced to the \v Cech complex by a
sequence of collapses, assuming some sampling conditions are
met. This result shows that it suffices to study \v Cech
  complexes.

In this paper, we give some mild conditions under which there is a
sequence of collapses that converts the \v Cech complex (and therefore
also the Rips complex) into a simplicial complex homeomorphic to the
shape. Our result assumes the shape to be a subset of the
$d$-dimensional Euclidean space with the property to be {\em nicely
  triangulable}, a concept we will explain later in the paper.

Perhaps unfortunately, our proof that a sequence of collapses exists
is not very constructive: it starts by sweeping space with offsets of
the shape --- which are unknown --- and builds a sequence of complexes
which have no reason to remain close to flag complexes and therefore
cannot benefit from the data structure developed in
\cite{ijcga11-blockers}. Nonetheless, even if results presented here
do not give yet any practical algorithm, we believe that they provide
a better understanding as to why the \v Cech complex (and therefore
the Rips complex) can be simplified by collapses and how this ability
is connected to the underlying metric structure of the space. In the
same spirit, we should mention \cite{adiprasito2011metric}, in which
the authors prove that every complex that is CAT(0) with a metric for
which all vertex stars are convex, is collapsible.

We now list the principal results of the paper, materialized as brown
arrows in Figure \ref{fig:diagram-results}. We also mention some
auxiliary results, which are interesting in their own rights. In
Section~\ref{section:restricted-reach}, we study how the reach of a
shape is modified when intersected with a (possibly infinite)
collection of balls and establish the contractibility of the
intersection, assuming the balls are small enough. In Section
\ref{section:restricted-cech}, we introduce the \v {\em Cech complex
  restricted by the shape $A$} and deduce conditions under which it is
homotopy equivalent to $A$ (Theorem
\ref{theorem:RestrictedCechCaptureHomotopyType}). In Section
\ref{section:restricting-the-cech}, we provide general conditions
under which the nerve of a family of compact sets undergoes collapses
as the compact sets evolve over time. Applying this technical result
to our context, we obtain conditions under which there is a sequence
of collapses that goes from the \v Cech complex to the restricted \v
Cech complex (Theorem
\ref{theorem:CechCollapsesOnRestrictedCech}). Combined with Theorem
\ref{theorem:RestrictedCechCaptureHomotopyType}, this gives an
alternative proof to a result
\cite{niyogi08:_findin_homol_of_subman_with} recalled in Section
\ref{section:background} (Lemma~\ref{lemma:CechReconstruit}).  In
Section \ref{section:collapsing-restricted-cech}, we define {\em
  $\alpha$-robust coverings} and give conditions under which the
restricted \v Cech complex can be transformed into the nerve of an
$\alpha$-robust covering (Theorem
\ref{theorem:collapsing-towards-nice-triangulations}). In Section
\ref{section:nice-triangulations}, we define and study {\em nicely
  triangulable spaces}. Such spaces enjoy the property of having
triangulations that can be expressed as the nerve of $\alpha$-robust
coverings for a large range of $\alpha$. Finally, we provide examples
of such spaces. Our list includes affine subspaces of the
  $d$-dimensional Euclidean space. It also contains the 2-sphere, the
  flat torus and all surfaces $C^{1,1}$ diffeomorphic to these
  two. Although our list is quite short, it is conceivable that many
  more spaces could be added. Actually, we conjecture that all compact
  smooth manifolds embedded in $\R^d$ are nicely triangulable and
  leave open this conjecture for future research. Section
\ref{section:conclusion} concludes the paper.

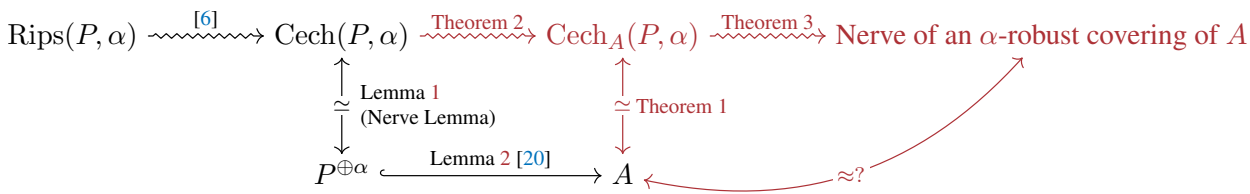
\begin{figure}[htb]
  \begin{center}
    \begin{tikzcd}[column sep = large, row sep = large]
      \Rips{P,\alpha} \arrow[rightsquigarrow]{r}{\text{\cite{Rips-cgta-AttaliLieutierSalinas}}} & \Cech{P,\alpha}
      \arrow[rightsquigarrow, color = Maroon]{r}{\text{Theorem \ref{theorem:CechCollapsesOnRestrictedCech}}}
      \arrow[leftrightarrow]{d}[description]{\simeq}{\; \bb{Lemma~\ref{lemma:nerve}}{(Nerve Lemma)}}
      & {\color{Maroon} \RCech{A}{P,\alpha}} \arrow[rightsquigarrow, color = Maroon]{r}{\text{Theorem \ref{theorem:collapsing-towards-nice-triangulations}}}
      \arrow[leftrightarrow, color = Maroon]{d}[description]{\simeq}{\;\text{Theorem
          \ref{theorem:RestrictedCechCaptureHomotopyType}}}
      & \text{\color{Maroon} Nerve of an $\alpha$-robust covering of $A$} \\
       & \Offset{P}{\alpha}
      \arrow[hookrightarrow]{r}{\text{Lemma
          \ref{lemma:CechReconstruit} \cite{niyogi08:_findin_homol_of_subman_with}}} &
      A \arrow[bend right,leftrightarrow, color = Maroon]{ur}[description]{\approx?} &
    \end{tikzcd}    
  \end{center}
  \caption{Logical structure of our results. Brown arrows represent new
    results. The arrow $\hookrightarrow$ stands for
    \enquote{deformation retracts to}. The arrow $\rightsquigarrow$
    stands for \enquote{can be transformed by a sequence of collapses
      into}. The symbol \enquote{$\simeq$} means \enquote{homotopy
      equivalent to} and \enquote{$\approx$} means
    \enquote{homeomorphic to}. \label{fig:diagram-results}}
\end{figure}


\section{Background}
\label{section:background}


First let us explain some of our terms and introduce the necessary
background. $\R^d$ denotes the $d$-dimensional Euclidean
space. $\|x-y\|$ is the Euclidean distance between two points $x$ and
$y$ of $\R^d$. The closed ball with center $x$ and radius $r$ is
denoted by $\Ball{x}{r}$ and its interior by $\OpenBall{x}{r}$. Given
a subset $X \subset \R^d$, the $\alpha$-offset of $X$ is
$\Offset{X}{\alpha} = \bigcup_{x \in X} B(x,\alpha)$. The Hausdorff
distance $d_H(X,Y)$ between the two compact sets $X$ and $Y$ of $\R^d$
is the smallest real number $\varepsilon \geq 0$ such that $X \subset
\Offset{Y}{\varepsilon}$ and $Y \subset \Offset{X}{\varepsilon}$. We
write $d(x,Y) = \inf_{y \in Y} \|x - y\|$ for the distance between
point $x \in \R^d$ and the set $Y \subset \R^d$ and $d(X,Y) = \inf_{x
  \in X} \inf_{y \in Y} \|x-y\|$ for the distance between the two sets
$X \subset \R^d$ and $Y\subset \R^d$.

A convenient way to build a simplicial complex is to consider the
nerve of a collection of sets. Specifically, let $P$ be a set of
indices. Later on, elements of $P$ will be points in $\R^d$. Let
$\Cset = \{ C_p \mid p\in P\}$ be a family of sets indexed by $p \in
P$. The nerve of the family is the abstract simplicial complex that
consists of all non-empty finite subcollections whose sets have
a non-empty common intersection. Formally, writing $\card
  \sigma$ for the number of elements in $\sigma$, we have
  $\Nerve{\Cset} = \{ \sigma \subset P \mid 0 < \card \sigma < +\infty
  \text{ and } \bigcap_{p \in \sigma} C_p \neq \emptyset\}$.  In this
paper, we shall consider nerves of coverings of a shape $A$.  We
recall that a {\em covering} of $A$ is a collection $\Cset = \{ C_p
\mid p \in P\}$ of subsets of $A$ so that $A = \bigcup_{p \in P}
C_p$. It is a {\em closed} ({\em resp. compact}) covering if each set
in $\Cset$ is closed ({\em resp. } compact). It is a {\em finite}
covering if the set of indices $P$ is finite. The Nerve Lemma gives a
condition under which the nerve of a covering of a shape shares the
topology of the shape. It has several versions
\cite{bjorner1996topological} and we shall use the following form:

\begin{lemma}[Nerve Lemma]
\label{lemma:nerve}
  Consider a compact set $A \subset \R^d$. Let $\Cset = \{C_p \mid p
  \in P\}$ be a finite closed covering of $A$. If for every $\emptyset
  \neq \sigma \subset P$, the intersection $\bigcap_{z \in \sigma}
  C_z$ is either empty or contractible, then the underlying space of
  $\Nerve{\Cset}$ is homotopy equivalent to $A$.
\end{lemma}

Hereafter, we shall omit the phrase \enquote{the underlying space of}
and write $X \simeq Y$ to say that $X$ is homotopy equivalent to
$Y$. Given a finite set of points $P \in \R^d$ and a real number
$\alpha \geq 0$, the \v Cech complex of $P$ at scale $\alpha$ can be
defined as $\Cech{P,\alpha} = \Nerve{ \{ B(p,\alpha) \mid p \in P\}
}$. With this definition and the Nerve Lemma, it is clear that
$\Cech{P,\alpha} \simeq \Offset{P}{\alpha}$; see the black vertical
arrow in Figure \ref{fig:diagram-results}. Several recent results have
expressed conditions under which $\Offset{P}{\alpha}$ recovers the
homotopy type of the shape $A$
\cite{niyogi08:_findin_homol_of_subman_with,chazal08:_smoot_manif_recon_from_noisy,chazal2009sampling,Rips-cgta-AttaliLieutierSalinas}. Intuitively,
the data points $P$ must sample the shape $A$ sufficiently densely and
accurately. One of the simplest ways to measure the quality of the
sample is to use the reach of the shape. Given a compact subset $A$ of
$\R^d$, recall that the {\em medial axis} $\MA{A}$ of $A$ is the set
of points in $\R^d$ which have at least two closest points in $A$. The
{\em reach} of $A$ is the infimum of distances between points in $A$
and points in $\MA{A}$, $\Reach{A} = \inf_{a \in A, m \in \MA{A}}
\|a-m\|$. It is well-known that a compact subset $C \subset \R^d$ has
  infinite reach if and only if $C$ is convex.  We have
(horizontal black arrow in Figure \ref{fig:diagram-results}):


\begin{lemma}[\cite{niyogi08:_findin_homol_of_subman_with}]
\label{lemma:CechReconstruit}
Let $A$ and $P$ be two compact subsets of $\R^d$. Suppose there
exists a real number $\varepsilon$ such that
$d_H(A,P) \leq \varepsilon < (3 - \sqrt{8}) \Reach{A}$. 
Then, $\Offset{P}{\alpha}$ deformation retracts to
$A$ for $\alpha = (2+\sqrt{2})\varepsilon$.
\end{lemma}

Combining Lemma~\ref{lemma:nerve} and
Lemma~\ref{lemma:CechReconstruit}, we thus get conditions under which
$\Cech{P,\alpha} \simeq A$.  The next two sections will provide an
alternative proof of this result along the way. We recall a
  result which will be useful when establishing some of the
  intermediate geometric lemmas. For a point $x \in A \setminus
  \MA{A}$, write $\pi_A(x)$ for the unique point in $A$ closest to
  $x$. We have:

\begin{lemma}[{\cite[p.~305]{delfour2011shapes}}]
\label{lemma:ProjectionWithinReachDistance}
Let $A$ be a compact subset of $\R^d$ and $c$ a point such that $0<
d(c,A) < \Reach{A}$.  Let $\Delta_c$ be the half-line with end point
$\pi_A(c)$ and containing $c$. For every point $x \in \Delta_c$, if
$d(x, \pi_A(c) ) < \Reach{A}$, then $\pi_A(x) = \pi_A(c)$.
\end{lemma}

Before starting the paper, we recall that a {\em collapse} of an
abstract simplicial complex $K$ is the removal of a simplex
$\sigma_{\min} \in K$ together with all its cofaces assuming
$\sigma_{\min}$ is non-maximal and its set of cofaces contains a
unique maximal element $\sigma_{\max} \in K$. A collapse produces a
simplicial complex to which $K$ deformation retracts and thus is a
simplification operation that preserves the homotopy type
\cite{cohen1973course}.


\section{Reach of spaces restricted by small balls}
\label{section:restricted-reach}


In this section, we consider a subset $A \subset \R^d$ such that
$\Reach{A}>0$ and prove that as we intersect $A$ with balls of radius
$\alpha < \Reach{A}$ the reach of the intersection can only get
bigger; see Figure \ref{fig:reach}. More precisely, let $A$ be a
compact subset of $\R^d$ and $\sigma \subset \R^d$. Write
$\CapBalls{\sigma}{\alpha} = \bigcap_{z \in \sigma} B(z,\alpha)$ for
the common intersection of balls with radius $\alpha$ centered at
$\sigma$ and assume that $A \cap \CapBalls{\sigma}{\alpha} \neq
\emptyset$. In this section, we establish that $\Reach{A} \leq
\Reach{A \cap \CapBalls{\sigma}{\alpha}}$ in the following
  situations: first when $\sigma$ is reduced to a single point $z$
(Lemma~\ref{lemma:InterWithBallKeepReach}), then when $\sigma$ is
finite (Lemma~\ref{lemma:InterWithFiniteBallsKeepReach}) and finally
when $\sigma$ is a compact subset of $\R^d$
(Lemma~\ref{lemma:InterWithBallsCollectionKeepReach}). Although the
first generalization ($\sigma$ finite) is all we need for establishing
Theorem \ref{theorem:RestrictedCechCaptureHomotopyType} in Section
\ref{section:restricted-cech}, the second generalization ($\sigma$
compact) will turn out to be useful later on in the paper. Let
  us start with a preliminary geometric lemma:

\begin{figure}[h]
\def\svgwidth{0.9\linewidth}
  \centering\small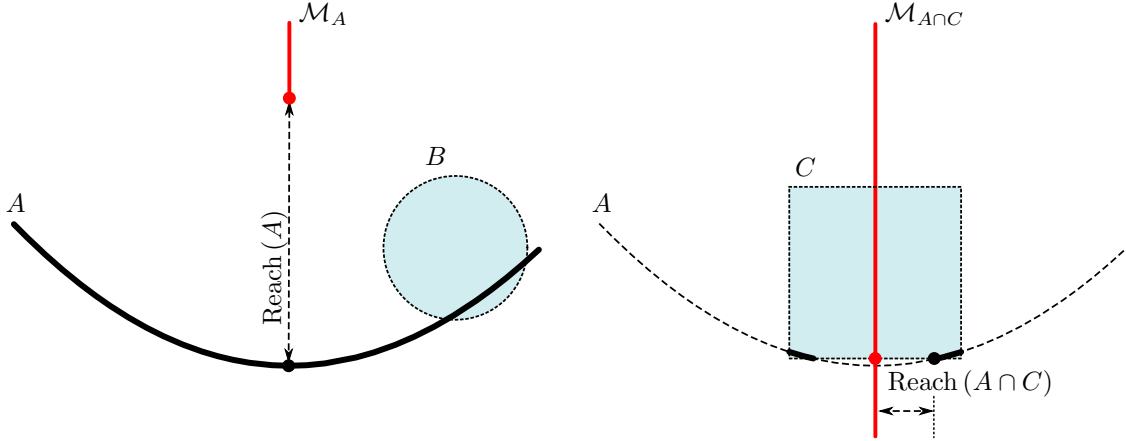
  \caption{Left: Medial axis $\MA{A}$ of a shape $A$. If we intersect
    $A$ with a ball $B$ whose radius is smaller than the reach of $A$,
    the reach of the intersection $A \cap B$ can only get
    bigger. Right. This property does not hold if we replace the ball
    $B$ by another set, even with infinite reach such as the solid
    cube $C$. \label{fig:reach}}
\end{figure}

\begin{lemma}\label{lemma:lemma0}
Let $X \subset \R^d$ be a non-empty compact set and  $B(c,\rho)$ 
its smallest enclosing ball. For all points $z \in \R^d$ and all real
numbers $r \geq \rho$, the following implications hold
\begin{enumerate}[{\em (i)}]
\item $X \subset \Ball{z}{r} \implies \Ball{c}{r - \sqrt{r^2 - \rho^2}} \subset \Ball{z}{r}$;
\item $\Ball{c}{r - \sqrt{r^2 - \rho^2}} \subset \OpenBall{z}{r} \implies
  X \cap \OpenBall{z}{r} \neq \emptyset$.
\end{enumerate}
\end{lemma}
\begin{proof}

To establish (i), assume for a contradiction that $\Ball{z}{r}$
does not contain $\Ball{c}{r - \sqrt{r^2 - \rho^2}}$ or
equivalently that $\|c-z\| > \sqrt{r^2 - \rho^2}$.  This implies
that the smallest ball enclosing $\Ball{c}{\rho} \cap
\Ball{z}{r}$ has radius less than $\rho$. Since this intersection
contains $X$, this would contradict the definition of $\rho$ as the
radius of the smallest ball enclosing $X$.

It is not hard to check (by contradiction) that the center of the
smallest ball enclosing $X$ lies on the convex hull of $X$. It follows
that for any half-space $H$ whose boundary passes through $c$, the
intersection $X \cap H \cap B(c,\rho)$ is non-empty. To establish
(ii), we may assume that $c \neq z$ for otherwise the result is
clear. Let $H$ be the half-space containing $z$ whose boundary passes
through $c$ and is orthogonal to the segment $cz$; see Figure
\ref{fig:proof-technical-1}. If $\Ball{c}{r - \sqrt{r^2 -
    \rho^2}} \subset \OpenBall{z}{r}$, then $H \cap B(c,\rho)
\subset \OpenBall{z}{r}$. Since $X \cap H \cap B(c,\rho) \neq
\emptyset$, it follows that $X \cap \OpenBall{z}{r} \neq
\emptyset$.
\end{proof}

\begin{figure}[htb]
  \def\svgwidth{0.4\linewidth}
  \centering\small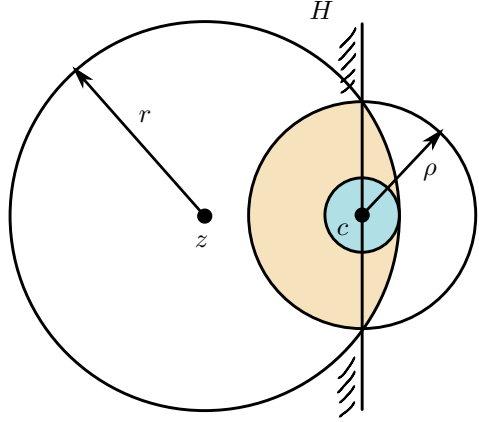
  \caption{Notation for the proof of Lemma~\ref{lemma:lemma0} 
    when $\|c-z\| = \sqrt{r^2 - \rho^2}$.
    \label{fig:proof-technical-1}}
\end{figure}

\begin{lemma}\label{lemma:InterWithBallKeepReach}
Let $A \subset \R^d$ be a compact set and $B(z,\alpha)$ a closed ball
with center $z$ and radius $\alpha$.  If $0 \leq \alpha < \Reach{A}$
and $A \cap B(z,\alpha) \neq \emptyset$ then $ \Reach{A} \leq \Reach{A
  \cap B(z,\alpha)}$.
\end{lemma}

\begin{figure}[htb]
  \def\svgwidth{1\linewidth}
  \centering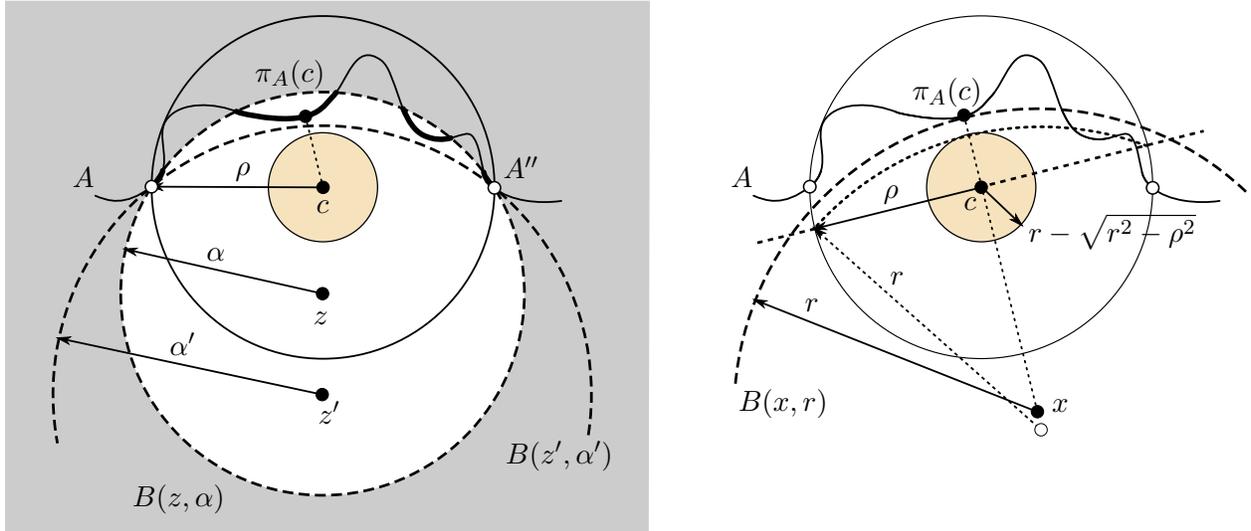
  \caption{Notation for the proof of
    Lemma~\ref{lemma:InterWithBallKeepReach}. The quantity $r -
    \sqrt{r^2 - \rho^2}$ represents the height of a spherical cap
    whose base has radius $\rho$ and which lies on a sphere with
    radius $r$.
    \label{fig:proof-reach-increasing}
  }
\end{figure}

\begin{proof}
See Figure \ref{fig:proof-reach-increasing} on the left. Assume, by
contradiction, that $\Reach{A \cap B(z,\alpha)} < \Reach{A}$ and
consider a point $z'$ in the medial axis of $A \cap B(z,\alpha)$ such
that 
$$
d(z',A \cap B(z,\alpha)) = \alpha' < \Reach{A}.
$$ Introduce $A'' = A \cap B(z,\alpha) \cap B(z',\alpha')$ and denote
by $c$ and $\rho$ the center and the radius of the smallest ball
enclosing $A''$. Because $A''$ is contained in both $B(z,\alpha)$ and
$B(z',\alpha')$, the radius of the smallest ball enclosing $A''$
satisfies $\rho \leq \min\{\alpha,\alpha'\} < \Reach{A}$.  Because
$A'' \subset A$, we get $d(c,A) \leq d(c,A'') \leq \rho < \Reach{A}$
and therefore $c$ has a unique closest point $\pi_A(c)$ in $A$.
Take $r$ to be any real number such that
$\max\{\alpha,\alpha'\} < r < \Reach{A}$. We claim that $r -
\sqrt{r^2 - \rho^2} < d(c,A)$; see Figure
\ref{fig:proof-reach-increasing} for a geometric interpretation of the
quantity $r - \sqrt{r^2 - \rho^2}$. 
Indeed, for every $(z_0,r_0) \in \{(z,\alpha),(z',\alpha')\}$,
since $A'' \subset B(z_0,r_0)$ and $r_0 \geq \rho$,
Lemma~\ref{lemma:lemma0} (i) implies that the following inclusion
holds:
$$
\Ball{c}{r_0 - \sqrt{r_0^2 - \rho^2}} ~~\subset~~ \Ball{z_0}{r_0}.
$$
Since the map $r \mapsto r - \sqrt{r^2 - \rho^2}$ is
strictly decreasing in $[\rho, +\infty)$ and $\rho \leq r_0 < r$, we get that
$$
\Ball{c}{r - \sqrt{r^2 - \rho^2}} ~~\subset~~
\OpenBall{z}{\alpha} \cap \OpenBall{z'}{\alpha'}.
$$ By construction, $\OpenBall{z'}{\alpha'}$ contains no points of $A
\cap \OpenBall{z}{\alpha}$. It follows that $\OpenBall{z}{\alpha} \cap
\OpenBall{z'}{\alpha'}$ contains no points of $A$, and neither does
$\Ball{c}{r - \sqrt{r^2 - \rho^2}}$.  Thus, $r - \sqrt{r^2 - \rho^2} <
d(c,A) = \|c - \pi_A(c)\|$ as claimed.  Let us consider the point $x =
\pi_A(c) + \frac{r}{d(c,A)} (c- \pi_A(c))$; see Figure
\ref{fig:proof-reach-increasing} on the right.  By construction
$\|x-\pi_A(c)\| = r < \Reach{A}$ and therefore $x$ has a unique
closest point $\pi_A(x)$ in $A$ which, by
Lemma~\ref{lemma:ProjectionWithinReachDistance}, satisfies $\pi_A(x) =
\pi_A(c)$. Since $\Ball{c}{d(c,A)}\subset\Ball{x}{d(x,A)}$, we deduce
that
$$
\Ball{c}{r- \sqrt{r^2 -\rho^2}} ~~\subset~~ \OpenBall{x}{r}.
$$ Applying Lemma~\ref{lemma:lemma0} (ii) we get that
$A'' \cap \OpenBall{x}{r} \neq \emptyset$ and therefore $\Ball{x}{r}$
contains points of $A$ in its interior. But this contradicts $d(x,A) =
r$.
\end{proof}

\begin{lemma}\label{lemma:InterWithFiniteBallsKeepReach}
  Consider a compact set $A \subset \R^d$ and a finite set $\sigma
  \subset \R^d$.  If $0 \leq \alpha < \Reach{A}$ and $A \cap
  \Bset(\sigma,\alpha) \neq \emptyset$ then $\Reach{A} \leq \Reach{A
    \cap \Bset(\sigma,\alpha)}$.
\end{lemma}

\begin{proof}
  By induction over the size of $\sigma$. 
\end{proof}

The following lemma is a milestone for the proof of Lemma~\ref{lemma:InterWithBallsCollectionKeepReach}.

\begin{lemma}\label{lemma:ReachIsPreservedByCountableIntersection}
Let $(A_n)_{n \in \N}$ be a sequence of non-empty compact
subsets of $\R^d$ decreasing with respect to the inclusion order. If there exists
a real number $r$ such that $0 \leq r \leq \Reach{A_n}$ for all $n \in \N$,
then $r \leq \Reach { \bigcap_{n\in \N} A_n }$.
\end{lemma}

\begin{proof}
Letting $A=\bigcap_{n\in \N} A_n$, we first show that the Hausdorff
distance $d_H(A_n,A)$ tends to $0$ as $n \to +\infty$. For
$\varepsilon > 0$, introduce the set $L_\varepsilon = \{ x \in \R^d
\mid d(x, A) \geq \varepsilon \}$ and notice that $\bigcap_{n \in \N}
\left( L_\varepsilon \cap A_n \right) = L_\varepsilon \cap A =
\emptyset$. Since the sequence of compact sets $(L_\varepsilon \cap
A_n)_{n \in \N}$ is decreasing, the only possibility is that
$L_\varepsilon \cap A_i = \emptyset$ for some $i \in
\N$. Equivalently, $d_H(A_i, A)< \varepsilon$ which proves the
convergence of $A_n$ to $A$ under Hausdorff distance.

Let $r'$ be a positive real number in the open interval $(0,r)$ and
let $z'$ be a point whose distance to $A$ is $r'$. Let us prove that
$z'$ has a unique closest point in $A$.  For $n$ large enough,
$d(z',A_n) < r$ and $z'$ has a unique closest point $a_n$ in
$A_n$. All points $a_n$ are contained in the closed ball $B(z',r)$ and
therefore, we can extract from $(a_n)_{n \in \N}$ a subsequence
$(a_{n_i})_{i \in \N}$ that converges to a point $a$. Since
$d_H(A_{n_i},A)$ tends to 0 as $i \to +\infty$, we deduce that the
point $a$ must belong to $A$. Let us define the point $z_{n_i}$ by
$$
z_{n_i} = a_{n_i} + \frac{r}{\|z'-a_{n_i}\|} (z'-a_{n_i}).
$$ The sequence $(z_{n_i})$ converges to the point $z = a +
\frac{r}{\|z'-a\|}(z'-a)$ and we get $d(z,A) \leq \|z-a\| = r$.  By
Lemma \ref{lemma:ProjectionWithinReachDistance}, the point $z_{n_i}$
shares with $z'$ the same closest point in $A_{n_i}$, namely $a_{n_i}$
and by construction $\|z_{n_i} - a_{n_i}\| = r$. Thus, $\OpenBall{z_{n_i}}{r}
\cap A = \emptyset$ and by passing to the limit, we get that
$\OpenBall{z}{r} \cap A = \emptyset$, or equivalently that $d(z,A)
\geq r$. Thus, $d(z,A) = r = \|z-a\|$ and since $z'$ lies on the open
line segment $za$, it has a unique closest point in $A$, namely the
point $a$. Since this is true for all $r' \in (0,r)$, we deduce that
$\Reach{A} \geq r$.
\end{proof}

\begin{lemma}\label{lemma:InterWithBallsCollectionKeepReach}
Let $A$ and $\sigma \neq \emptyset$ be two compact sets of $\R^d$.  If $0 \leq \alpha
< \Reach{A}$ and $A \cap \Bset(\sigma,\alpha) \neq \emptyset$ then
$\Reach{A} \leq \Reach{A \cap \Bset(\sigma,\alpha)}$.
\end{lemma}

\begin{proof}
Consider a sequence $\{z_i\}_{i\in \N}$ which is dense in $\sigma$.
Lemma~\ref{lemma:InterWithFiniteBallsKeepReach} implies that for
all $n \geq 0$, we have $\Reach{A \cap \bigcap_{i=1}^n
  \Ball{z_i}{\alpha}} \geq \alpha$. Applying Lemma~\ref{lemma:ReachIsPreservedByCountableIntersection} we get that
$\Reach{ A \cap \bigcap_{i \in \N} \Ball{z_i}{\alpha} } \geq \alpha$.
We claim that $A \cap \bigcap_{i \in \N} \Ball{z_i}{\alpha} = A \cap
\bigcap_{z\in \sigma}\Ball{z}{\alpha}$. One direction is trivial. If
$a \in A \cap \bigcap_{z\in \sigma}\Ball{z}{\alpha}$, then $a \in A
\cap \bigcap_{i \in \N} \Ball{z_i}{\alpha}$. For the other direction,
take $a \in A \cap \bigcap_{i\in \N} \Ball{z_i}{\alpha}$ and let us
prove that $\forall z \in \sigma, \|z-a\| \leq \alpha$. Assume, by
contradiction, that for some $z \in \sigma$, one has $\|z-a\| >
\alpha$. Then, there is $i\in \N$ such that $\|z - z_i\| < \|z-a\| -
\alpha$, yielding $\|z_i - a \| \geq \|z-a\| - \|z-z_i\| > \alpha$,
which is impossible. We have just shown that $a \in A \cap
\bigcap_{z\in \sigma}\Ball{z}{\alpha}$.
\end{proof}


\section{The restricted \v Cech complex}
\label{section:restricted-cech}


Given a subset $A \subset \R^d$, a finite point set $P$ and a real
number $\alpha \geq 0$, let us define the \v Cech complex of $P$ at
scale $\alpha$, $\RCech{A}{P,\alpha}$, restricted by $A$ as the set of
simplices spanned by points in $P$ that fit in a ball of radius
$\alpha$ whose center belongs to $A$. Equivalently,
$\RCech{A}{P,\alpha} = \Nerve{ \{ A \cap \Ball{p}{\alpha} \mid p\in P
  \}}$. In this section, we give conditions under which $A$ and
$\RCech{A}{P,\alpha}$ are homotopy equivalent
(Theorem~\ref{theorem:RestrictedCechCaptureHomotopyType}).  Recall
that $\CapBalls{\sigma}{\alpha} = \bigcap_{z \in \sigma} B(z,\alpha)$
is the common intersection of balls with radius $\alpha$ centered at
$\sigma$. In the proof, we will argue that $A \cap
  \CapBalls{\sigma}{\alpha}$ is either empty or contractible, whenever
  $0 \leq \alpha < \Reach{A}$. This argument is encapsulated in Lemma
  \ref{lemma:topological} and follows from Lemma \ref{lemma:ReachLargerThanRadiusThenContractible}. Let $\Radius{X}$ designate
the radius of the smallest ball enclosing the compact set $X$.

\begin{lemma}\label{lemma:ReachLargerThanRadiusThenContractible}
If $X \subset \R^d$ is a non-empty compact set with
$\Radius{X}<\Reach{X}$, then $X$ is contractible.
\end{lemma}

\begin{figure}[htb]
  \def\svgwidth{0.35\linewidth}
  \centering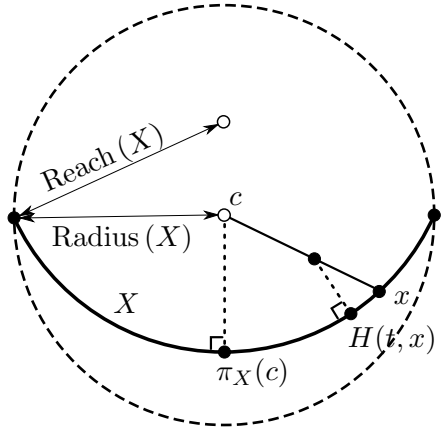
  \caption{Notation for the proof of Lemma~\ref{lemma:ReachLargerThanRadiusThenContractible}.
    \label{fig:proof-reach-larger-than-radius-then-contractible}
  }
\end{figure}


\begin{proof}
We recall that for every point $m$ such that $d(m,X) < \Reach{X}$
there exists a unique point of $X$ closest to $m$, which we denote by
$\pi_X(m)$. Furthermore, we know from \cite[page 435]{federer-59} that
for $0 < r < \Reach{X}$ the projection map $\pi_X$ onto $X$ is
Lipschitz for points at distance less than $r$ from $X$. Denote by $c$
the center of the smallest ball enclosing $X$; see Figure
\ref{fig:proof-reach-larger-than-radius-then-contractible}.  If
$x\in X$ and $t \in [0,1]$, one has
$$ d( (1-t) x +t c, X) ~~\leq~~ \| (1-t) x + t c - x\| ~~\leq~~ \|c-x\| ~~\leq~~
\Radius{X} ~~<~~ \Reach{X}.$$
Thus, the map $H : [0,1] \times X
\rightarrow X$ defined by $H(t,x) = \pi_X((1-t) x +t c)$ is Lipschitz
and defines a deformation retraction of $X$ onto $\{\pi_X(c) \}$.
\end{proof}

We deduce immediately the following lemma. Besides being useful for
proving Theorem~\ref{theorem:RestrictedCechCaptureHomotopyType}, it
will turn out to be a key tool in Section
\ref{section:collapsing-restricted-cech}.

\begin{lemma}
  \label{lemma:topological}
  Let $A$ be a compact set of $\R^d$ and $\alpha$ a real number such
  $0 \leq \alpha < \Reach{A}$. For all non-empty compact subsets $\sigma
  \subset \R^d$, the intersection $A \cap \Bset(\sigma,\alpha)$ is
  either empty or contractible.
\end{lemma}

\begin{proof}
Suppose $A \cap \CapBalls{\sigma}{\alpha} \neq \emptyset$. By
Lemma~\ref{lemma:InterWithBallsCollectionKeepReach},
$$\Radius{A \cap
  \CapBalls{\sigma}{\alpha}} ~~\leq~~ \alpha ~~<~~ \Reach{A} ~~\leq~~ \Reach{A \cap
  \CapBalls{\sigma}{\alpha}}.$$
By Lemma~\ref{lemma:ReachLargerThanRadiusThenContractible}, $A \cap
\CapBalls{\sigma}{\alpha}$ is contractible.
\end{proof}

This lemma can be seen as a variant of Lemma~7 in
\cite{amenta1999surface}, Proposition 12 in
\cite{boissonnat2001natural} and the local reach lemma in \cite{spm07}
which all say that if $A$ is a $k$-manifold that intersects a ball $B$
with radius $\alpha <\Reach{A}$, then $A \cap B$ is a topological
$k$-ball. 

\begin{theorem}\label{theorem:RestrictedCechCaptureHomotopyType}
Let $A \subset \R^d$ be a compact set, $P \subset \R^d$ a finite point
set and $\alpha$ a real number such that $0 \leq \alpha < \Reach{A}$
and $A \subset \Offset{P}{\alpha}$. Then, $\RCech{A}{P,\alpha}$ and
$A$ have the same homotopy type.
\end{theorem}

\begin{proof}
  Since $A \subset \Offset{P}{\alpha}$, clearly $A = \bigcup_{p \in P}
  (A \cap B(p,\alpha))$. By Lemma~\ref{lemma:topological}, for
  all $\emptyset \neq \sigma \subset P$, the intersection $\bigcap_{z
    \in \sigma}(A \cap B(z,\alpha))$ is either empty or
  contractible. We conclude by applying the Nerve Lemma to the
  collection $\{A \cap B(p,\alpha) \mid p \in P\}$.
\end{proof}


\section{Restricting the \v Cech complex by collapses}
\label{section:restricting-the-cech}


In this section, we state our second theorem (horizontal brown
  arrow in Figure \ref{fig:diagram-results}). The theorem describes
  conditions under which there exists a sequence of collapses that
  transforms $\Cech{P,\alpha}$ into its restricted version
  $\RCech{A}{P,\alpha}$. It can be seen as a combinatorial version of
  Lemma~\ref{lemma:CechReconstruit} which says that, under the same
  hypotheses, there is a deformation retraction of
  $\Offset{P}{\alpha}$ onto $A$. Instrumental in proving the theorem,
  we need several facts about the distance between a collection of
  balls and a shape $A$. These facts are formalized in Lemma
\ref{lemma:UnicityOfClosesPointInBallsIntersection}. As before, we let
$\CapBalls{\sigma}{\alpha}$ denote the common intersection of balls
with radius $\alpha$ centered at $\sigma$ and by convention, we set
$d(A,\emptyset) = +\infty$. Hence, when we write that
$d(A,\CapBalls{\sigma}{\alpha}) = t$ for some $t \in \R$, this implies
implicitly that $\CapBalls{\sigma}{\alpha} \neq \emptyset$.

\begin{lemma}\label{lemma:UnicityOfClosesPointInBallsIntersection}
Let $A \subset \R^d$ be a compact set, $\sigma \subset \R^d$ a finite
set and $\alpha \geq 0$ such that $d(A,\CapBalls{\sigma}{\alpha}) = t$ for some $t \in \R$. If $0 < t < \Reach{A} -
\alpha$, we have the following properties (see Figure \ref{fig:proof-restricting-the-cech}, left):
\begin{itemize}\denselist
\item There exists a unique point $x \in \CapBalls{\sigma}{\alpha}$ whose distance to $A$ is $t$ ;
\item The set $\sigma_0 = \{ p \in \sigma \mid x \in \partial B(p,\alpha) \}$ is
  non-empty ;
\item $d(A,\CapBalls{\sigma_0}{\alpha}) = t$.
\end{itemize}
\end{lemma}

\begin{figure}[htb]
  \def\svgwidth{0.9\linewidth}
  \centering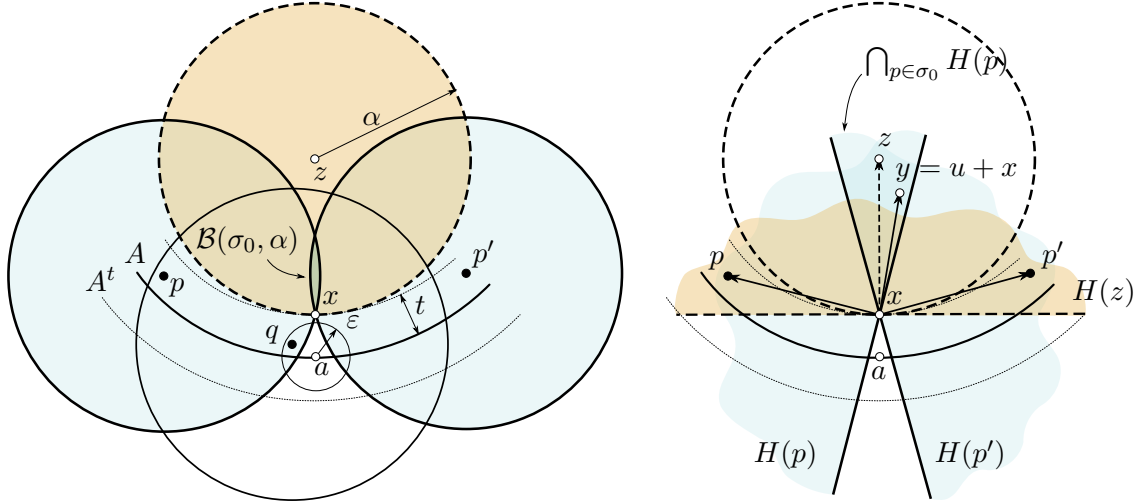
  \caption{ Notation for the proofs of
    Lemma~\ref{lemma:UnicityOfClosesPointInBallsIntersection} and
    Theorem \ref{theorem:CechCollapsesOnRestrictedCech}. Black dots
    belong to $\sigma$ and the ball $B(z,\alpha)$ is instrumental in
    proving
    Lemma~\ref{lemma:UnicityOfClosesPointInBallsIntersection}. We show
    that $\CapBalls{\sigma_0}{\alpha} \subset B(z,\alpha)$ (on the
    left) using the fact that $\bigcap_{p \in \sigma_0} H(p) \subset
    H(z)$ (on the right), where $H(m)$ designates the half-space
    which contains $B(m,\|x-m\|)$ and whose boundary passes
      through $x$.
    \label{fig:proof-restricting-the-cech}
  }
\end{figure}

\begin{proof}
Since $A$ and $\CapBalls{\sigma}{\alpha}$ are both compact sets, there
exists at least one pair of points $(a,x) \in A \times
\CapBalls{\sigma}{\alpha}$ such that $\|a-x\| = t$; see Figure
\ref{fig:proof-restricting-the-cech}. By definition, $x$ belongs to
$B(p,\alpha)$ for all $p \in \sigma$. Since $x$ lies on the boundary
of $\CapBalls{\sigma}{\alpha}$, it lies on the boundary of
$\Ball{p}{\alpha}$ for at least one $p \in \sigma$, showing that
$\sigma_0 \neq \emptyset$.  By construction, $a$ is the point of
  $A$ closest to $x$ and, by Lemma
  \ref{lemma:ProjectionWithinReachDistance}, it is also the
  point of $A$ closest to point $z = x + \alpha \frac{x-a}{\|x-a\|}$. It
  follows that $d(A,\{z\}) = \|a-z\| = \|a-x\| + \|x-z\| = t+\alpha$ is realized by the pair of
  points $(a,z)$ and the distance $d(A,B(z,\alpha)) = t$ is realized
by the pair of points $(a,x)$. To prove that $x$ is the unique point
of $\CapBalls{\sigma}{\alpha}$ whose distance to $A$ is $t$, it
suffices to show that $\CapBalls{\sigma}{\alpha} \subset
\Ball{z}{\alpha}$. Actually, we will prove a stronger result, namely
that $\CapBalls{\sigma_0}{\alpha} \subset \Ball{z}{\alpha}$ which will
also imply the third item of the lemma, that is,
$d(A,\CapBalls{\sigma_0}{\alpha}) = t$.

Let us associate to every point $m \in \R^d$ the closed half-space
$H(m)$ whose boundary passes through $x$ and which contains the ball
$B(m,\|m-x\|)$:
$$
H(m) = \{ y \in \R^d , \, \scalprod{m-x}{y-x} \;\geq \; 0 \}.
$$
We establish the following four statements:
\begin{enumerate}[{(A)}]
\item $\bigcap_{p \in \sigma} B(p,\alpha) \subset H(z)$;
\item $\bigcap_{p \in \sigma} H(p) \subset H(z)$;
\item $\bigcap_{p \in \sigma_0} H(p) \subset H(z)$;
\item $\bigcap_{p \in \sigma_0} B(p,\alpha) \subset B(z,\alpha)$;
\end{enumerate}

To establish (A), we note that, by construction, $x$ is the point of
$\CapBalls{\sigma}{\alpha}$ whose distance to $a$ is smallest and
therefore $\CapBalls{\sigma}{\alpha} \cap B(a,t) = \{x\}$ from which
statement (A) follows by convexity of
$\CapBalls{\sigma}{\alpha}$. Indeed, if there were $y \in
\CapBalls{\sigma}{\alpha} \setminus H(z)$, then $y \neq x$ and the
segment $xy$ would intersect the interior of $B(a,t)$. But, this is
impossible since $xy$ is contained in $\CapBalls{\sigma}{\alpha}$ and
$\CapBalls{\sigma}{\alpha}$ does not intersect the interior of
$B(a,t)$.  To prove (A) $\Rightarrow$ (B), we apply to the two sets on
both sides of (A) an homothety with center $x$ and ratio $s$. Consider
the half-line with origin at $x$ and passing through $p$ and let $p_s$
be the point on this half-line whose distance to $x$ is $s \geq
0$. Clearly, the image of the left side is $\bigcap_{p \in \sigma}
B(p_s,s\alpha)$ and the image of the right side is $H(z)$.  We thus
get that $\bigcap_{p \in \sigma} B(p_s,s\alpha) \subset H(z)$ for all
$s \geq 0$. Taking the limit as $s$ tends to infinity (or
equivalently, taking the union of left sides for all values of $s$),
we get (B).  Statement (B) means that for all $y \in \R^d$, the
following implication holds:
$$
\min_{p \in \sigma} \scalprod{p-x}{y-x} \geq 0 \implies \scalprod{z-x}{y-x} \geq 0.
$$ Noting that if the above implication holds for all $y$ in a small
neighborhood of $x$, then it holds for all $y \in \R^d$, we deduce that (B)
$\Rightarrow$ (C). To prove (C) $\Rightarrow$ (D), we observe that
(C) implies that for all $y \in \bigcap_{p\in \sigma_0} H(p)$, the
distance between $y$ and the boundary of $H(z)$ is always larger than
or equal to the distance between $y$ and the boundary of $H(p)$ for
some $p \in \sigma_0$. Formally, this means that for all $u \in \R^d$
and all $\delta \geq 0$, the following implication holds:
$$ 
\min_{p \in \sigma_0} \scalprod{p-x}{u} \, \geq \, \delta \; \implies \;
\scalprod{z-x}{u} \, \geq \, \delta
$$ Plugging $\delta = \frac{\|u\|^2}{2}$ in the above implication and
noting that for all $m \in \sigma_0 \cup \{z\}$, the following inequality
$2\scalprod{m-x}{u} \geq \|u\|^2$ can be rewritten as $\|(m-x)-u\|^2
\leq \alpha^2$, we get that for all $u \in \R^d$, the following implication
holds:
$$
\max_{p \in \sigma_0} \|p - u - x\|^2 \, \leq \, \alpha^2  \; \implies
\; \| z - u - x\|^2 \, \leq \, \alpha^2.
$$ Equivalently, (D) holds and $\CapBalls{\sigma_0}{\alpha} \subset
\Ball{z}{\alpha}$, as required.
\end{proof}

\begin{theorem}\label{theorem:CechCollapsesOnRestrictedCech}
Let $\varepsilon \geq 0$, $\alpha \geq 0$, and $r \geq 0$. Consider a
compact set $A\subset \R^d$ with $\Reach{A} \geq r$. Let $P \subset
\R^d$ be a finite set such that $d_H(A,P) < \varepsilon$. There
exists a sequence of collapses from $\Cech{P,\alpha}$ to
$\RCech{A}{P,\alpha}$ whenever $\varepsilon$, $\alpha$ and $r$ satisfy
the following two conditions:
\begin{enumerate}[{\em (i)}]
\item $\sqrt{2}\alpha < r - \varepsilon$;
\item $r - \sqrt{(r-\varepsilon)^2-\alpha^2} < \alpha- \varepsilon$.
\end{enumerate}
In particular, for $\varepsilon < (3 - \sqrt{8}) r$ and $\alpha = (2 +
\sqrt{2})\varepsilon$, conditions {\em(i)} and {\em (ii)} are fulfilled.
\end{theorem}

\begin{proof}
Letting $\beta = r - \sqrt{(r-\varepsilon)^2-\alpha^2}$, we observe
that condition (i) implies $\beta < r - \alpha$ and condition (ii) is
equivalent to $\beta < \alpha - \varepsilon$. For $t \geq 0$, we
define the simplicial complex $K_t = \Nerve{ \{ \Offset{A}{t} \cap
  \Ball{p}{\alpha} \mid p\in P \}}$.  Notice that $K_0 =
\RCech{A}{P,\alpha}$ and $K_{+\infty} = \Cech{P,\alpha}$.  Using the
fact that $K_t$ can equivalently be defined as $K_t = \{ \sigma
\subset P \mid d(A,\CapBalls{\sigma}{\alpha}) \leq t\}$, we deduce
that, as $t$ continuously decreases from $+\infty$ to 0, the complex
$K_t$ can only loose simplices and the set of simplices that disappear
at time $t$ is:
$$
\Delta_t = \{ \sigma \subset P \mid d(A, \CapBalls{\sigma}{\alpha}) = t \}.
$$


\paragraph{Generic case.} We first establish the theorem under the
following generic condition:
  \begin{enumerate}
  \item[$(\star)$] For all $s \in \R^+$, the set of simplices
    $\Delta_s$ is either empty or has a unique inclusion-minimal element.
  \end{enumerate}
At the end of the proof, we will explain what to do if the above
condition is not satisfied. Assuming we are in the generic case, we
proceed in two stages:

\medskip \noindent {\bf(a)} \hspace{2mm} First, we prove that $K_t$
does not change at all as $t$ decreases continuously from $+\infty$ to
$\beta$. In other words, $K_t = \Cech{P,t}$ for all $t \geq
\beta$. Note that this is equivalent to proving that for all non-empty
subsets $\sigma \subset P$ and all $t \geq \beta$,
$$
\bigcap_{p \in \sigma} \Ball{p}{\alpha} \neq \emptyset \iff 
\Offset{A}{t} \cap \bigcap_{p \in \sigma} \Ball{p}{\alpha} \neq \emptyset.
$$ One direction is trivial: if a point belongs to the intersection on
the right, then it belongs to the intersection on the left. If the
intersection on the left is non-empty, then it contains the center $z$
of the smallest ball enclosing $\sigma$ and $\Radius{\sigma} \leq
\alpha< \sqrt{2}\alpha < r - \varepsilon$. Lemma~14 in
  \cite{socg10-convex} states that if a subset $\sigma$ satisfies the
  following two conditions: (1) $\sigma \subset
  \Offset{A}{\varepsilon}$ and (2) $\Radius{\sigma} < r -
  \varepsilon$, then $\Conv{\sigma} \subset \Offset{A}{t}$ for all $t
  \geq \beta$. Since $z$ belongs to $\Conv{\sigma}$, it follows that
  $z$ also belongs to the $t$-offset $\Offset{A}{t}$ and therefore the
  intersection on the right is non-empty.

\medskip \noindent {\bf(b)} \hspace{2mm} Second, we prove that as $t$ decreases
continuously from $\beta$ to 0, the deletion of simplices $\Delta_t$
from $K_t$ is a collapse for all $t \in (0,\beta]$. Suppose $\Delta_t
\neq \emptyset$ for some $t \in (0,\beta]$ and let $\sigma_{\min}$ be
the unique inclusion-minimal element of $\Delta_t$. Since
$\sigma_{\min}$ disappears at time $t$, so does all its cofaces and it
follows that $\Delta_t$ is the set of cofaces of $\sigma_{\min}$.
Since $0 < t \leq \beta < r - \alpha \leq \Reach{A} - \alpha$,
Lemma~\ref{lemma:UnicityOfClosesPointInBallsIntersection} implies that
there exists a unique point $x \in \CapBalls{\sigma_{\min}}{\alpha}$
whose distance to $A$ is $t$; see Figure
\ref{fig:proof-restricting-the-cech}, left.  It is easy to see that
$\Delta_t$ has a unique inclusion-maximal element $\sigma_{\max} = \{
p \in P \mid x \in \Ball{p}{\alpha}\}$.  Thus, $\Delta_t$ consists of
all cofaces of $\sigma_{\min}$ and these cofaces are all faces of
$\sigma_{\max}$. To prove that removing $\Delta_t$ from $K_t$ is a
collapse, it suffices to establish that $\sigma_{\min} \neq
\sigma_{\max}$. By
Lemma~\ref{lemma:UnicityOfClosesPointInBallsIntersection}, we know
that $\sigma_0 = \{ p \in \sigma_{\min} \mid x \in \partial
B(p,\alpha)\}$ is non-empty and belongs to $\Delta_t$. By the choice
of $\sigma_{\min}$ as the minimal element of $\Delta_t$, we have
$\sigma_{\min} \subset \sigma_0$ and therefore $x$ lies on the
boundary of $B(p,\alpha)$ for all $p \in \sigma_{\min}$. Since $d(x,A)
= t \leq \beta < \Reach{A}$, there exists a unique point $a \in A$
such that $\|a-x\| = t$. Because $d_H(A,P) < \varepsilon$, we know
that there exists a point $q \in P$ such that $\|q-a\| \leq
\varepsilon$. Since $\|q-x\| \leq \|q-a\| + \|a-x\| \leq \varepsilon +
t \leq \varepsilon + \beta < \alpha$, we get that $x$ lies in the
interior of $B(q,\alpha)$. Therefore, $q$ belongs to $\sigma_{\max}$
but not to $\sigma_{\min}$. Hence, $\sigma_{\min} \neq \sigma_{\max}$.

\paragraph{Getting rid of the genericity assumption.}
We need first some definitions and notations. Given a collection of
maps $\xi_p : \R^+ \to \R^+$, one for each $p \in P$, we define the
simplicial complex
$$ K_t^\xi = \Nerve{ \{ \Offset{A}{\xi_p(t)} \cap \Ball{p}{\alpha}
  \mid p\in P \} }.$$ If each $\xi_p$ is an increasing continuous
bijection, the simplicial complex $K_t^\xi$ can only loose simplices
as $t$ continuously decreases from $+\infty$ to $0$. Precisely, the
set of simplices that disappear at time $t$ is:
$$
\Delta_t^\xi =
\{ \sigma \subset P \mid d(A, \CapBalls{\sigma}{\alpha} ) = \min_{p \in
  \sigma} \xi_p(t) \}.
$$ 
Given $\eta > 0$, we say that the map $\alpha : \R^+ \to \R^+$ is a
{\em standard $\eta$-perturbation} of the identity map if (1)
$\alpha$ is a continuous bijection; (2) $\alpha(0) = 0; (3) \lim_{p \to
  +\infty} \alpha(t) = +\infty$; (4) $t \leq \alpha(t) \leq t + \eta$
for all $t \in \R^+$. One can easily check that the composition of two
standard $\eta$-perturbations is a standard $(2\eta)$-perturbation.
Suppose now that each $\xi_p$ is a standard $\eta$-perturbation 
and notice that $K_0^\xi = \RCech{A}{P,\alpha}$ and
$K_{+\infty}^\xi = \Cech{P,\alpha}$.  By slightly adapting the first
part of the proof above, it is not difficult to establish that for
$\eta>0$ small enough, the simplicial complex $K_t^\xi$ only undergoes
collapses as $t$ continuously decreases from $+\infty$ to 0 under the
following generic condition:
\begin{enumerate}
\item[$(\star^\xi)$] For all $s \in \R^+$, the set of simplices
  $\Delta_s^\xi$ is either empty or has a unique inclusion-minimal
  element.
\end{enumerate}
We start by setting $\xi_p$ to the identity map for all $p \in P$.
If the generic condition $(\star^\xi)$ is not satisfied, we apply a
small perturbation to the maps $\xi_p$ so that after perturbation
the generic condition $(\star^\xi)$ is satisfied and each $\xi_p$ is
a standard $\eta$-perturbation. The construction
can be made so that $\eta>0$ is as small as desired and we can apply
our previous findings. For this, we proceed as
follows.  We say that two simplices $\sigma_1$ and $\sigma_2$ are {\em
  in conjunction at time $t$} if they are both inclusion-minimal
elements of $\Delta_t^\xi$ for some $t \in \R^+$. We say that $t$ is
an {\em event time} if $\Delta_t^\xi \neq \emptyset$. Consider two
simplices that are in conjunction at time $t$, say $\sigma_1$ and
$\sigma_2$. Suppose $q \in \sigma_1$ and $q \not \in
\sigma_2$. Consider an increasing continuous bijection $\psi : [0,1]
\to [0,1]$ that differs from identity only in a small neighborhood of
$t$ that does not include any other event times. Furthermore, we choose
$\psi$ such that $\psi(t) \geq t$. Replacing $\xi_q$ by $\xi_q \circ
\psi$ and leaving unchanged $\xi_p$ for all $p \in P \setminus
\{q\}$, we change the time at which $\sigma_1$ disappears while
keeping unchanged the time at which $\sigma_2$ disappears. After this
operation, $\sigma_1$ and $\sigma_2$ are not in conjunction
anymore. Furthermore, the operation does not create any new pair of
simplices in conjunction. By repeating this operation a finite number
of times, we thus get a new collection of maps $\xi_p$ as required.
\end{proof}


\section{Collapsing the restricted \v Cech complex}
\label{section:collapsing-restricted-cech}


In this section, we find conditions under which there is a sequence of
collapses transforming the restricted \v Cech complex
$\RCech{A}{P,\alpha}$ into the nerve of an $\alpha$-robust covering of
$A$. We define $\alpha$-robust coverings and state our result in
Section \ref{section:robust-coverings}. Our proof technique consists
in introducing a family of compact sets $\mathcal{D} = \{D_{p}(t) \mid
(p,t) \in P \times [0,1]\}$ and monitoring the evolution of its nerve as the
parameter $t$ increases continuously from $0$ to $1$. In Section
\ref{section:nerve-evolution}, we give some general conditions on
$\mathcal{D}$ that guarantee that the simplicial complex $K(t) =
\Nerve{ \{ D_p(t) \mid p \in P \}}$ only undergoes collapses as $t$
increases from 0 to 1. We believe that these conditions are
sufficiently general to be applied to other situations and therefore
are interesting in their own right. Armed with this tool, we establish
our third result in Section \ref{section:robust-coverings}, that is,
we find a family of compact sets $\mathcal{D}$ which enjoys the
properties required in Section \ref{section:nerve-evolution} and such
that $K(0) = \RCech{A}{P,\alpha}$ and $K(1)$ is isomorphic to the
nerve of an $\alpha$-robust covering of $A$.

\subsection{Evolving families of compact sets}
\label{section:nerve-evolution}

In this section, we present a tool that will be useful in the next
section for establishing 
Theorem~\ref{theorem:collapsing-towards-nice-triangulations}. Consider
a covering of a topological space and suppose this covering evolves
over time. We state conditions under which the evolution of the nerve
of this covering only undergoes collapses. Conditions are formulated
in a very general setting. We do not even need to endow the
topological space with a metric structure.  We only require the
topological space to be compact and $\operatorname{T}_1$
separable. Recall that a topological space $X$ is said to be {\em
  $\operatorname{T}_1$ separable} if for every pair of distinct points
$(a,b) \in X^2$, there exist two open sets $U_a$ and $U_b$ such that
$a \in U_a \setminus U_b$ and $b \in U_b \setminus U_a$. For instance,
metric spaces are $\operatorname{T}_1$ separable.

In Lemma~\ref{lemma:abstract-tool}, we will use the notion of
connectedness as defined in general topology: a topological space
({\em resp.} subspace) is {\em connected} if it cannot be represented as
the union of two disjoint non-empty open subsets ({\em resp.}
relatively open subsets).  Observe that if $X$ is a topological
$\operatorname{T}_1$ space, then for any point $a \in X$, the subspace
$X\setminus \{a\}$ is open. It follows that if $X$ is connected and
$X\setminus \{a\}$ is non-empty, then $\{a\}$ cannot be open and
$\{a\}^\circ= \emptyset$. Indeed, if $\{a\}$ were open, then $X = (X
\setminus \{a\}) \cup \{a\}$ would be expressed as the union of two
disjoint non-empty open subsets, a contradiction.

Before stating our lemma, let us introduce one additional
piece of notation. Given a finite set $\sigma$ and a map $\phi : \sigma
\to [0,1)$, we write $\phi' \succ \phi$ to designate a map $\phi' : \sigma \to
  [0,1]$ such that $\phi'(p) > \phi(p)$ for all $p \in
  \sigma$. We will say that the map $\phi$ is constant if
    $\phi(p)=\phi(q)$ for all $(p,q) \in \sigma^2$.

\begin{lemma}\label{lemma:abstract-tool}
Let $A$ be a compact topological $\operatorname{T}_1$ space and $P$ a
finite set. Consider a family of compact subsets of $A$, $\Dset = \{
D_p(t) \mid (p,t) \in P \times [0,1]\}$ which satisfies the following
five properties:
\begin{enumerate}[{\em (a)}]

\item For all $0 \leq t < t' \leq 1$ and all $p \in P$, we have
  $D_p(t') \subset D_p(t)^\circ$;

\item $\bigcup_{p \in P} D_p(1) = A$;

\item For all $\emptyset \neq \sigma \subset P$ and all maps $\phi :
  \sigma \to [0,1]$, the intersection $\Dset(\sigma,\phi) = \bigcap_{p
    \in \sigma} D_p \circ \phi(p)$ is either empty or connected;
  
\item For all $\emptyset \neq \sigma \subset P$ and all maps $\phi :
  \sigma \to [0,1)$, the following implication holds: $\Dset(\sigma,\phi)
  \neq \emptyset$ and $\Dset(\sigma,\phi') = \emptyset$ for all $\phi' \succ
  \phi$ implies that $\Dset(\sigma,\phi)$ is reduced to a single point.
  
\item For all $0 < \tau \leq 1$ and all $p \in P$, one has $D_p(\tau) =
  \bigcap_{t \in [0,\tau)} D_p(t)$ 

\end{enumerate}
Then, as $t$ increases continuously from $0$ to $1$, the simplicial
complex $K_t = \Nerve\{ D_p(t) \mid p\in P\}$ only undergoes collapses.
\end{lemma}

\begin{proof}
To prove the lemma, we may assume that $A$ is neither disconnected nor
reduced to a single point. Indeed, if $A$ is not connected then
condition (c) implies that for each $p \in P$, the subset
$D_p(0)$ is contained entirely within one connected component
of $A$ and the connected components of $A$ can be considered
separately. If $A$ is reduced to a single point, then the result is
clear.

Assuming $A$ is neither disconnected nor reduced to a single point, we
study the changes that occur in $K_t$ as $t$ increases continuously
from $0$ to $1$. Because of condition (a), some simplices may
disappear from $K_t$ but no simplices can ever appear in $K_t$. Given
a simplex $\sigma \in K_0 \setminus K_1$, we call $\tau_\sigma = \sup
\{ t \in [0,1] \mid \sigma \in K_t\}$ the {\em death time} of $\sigma$
and claim that $\sigma \in K_{\tau_\sigma}$. Indeed, if $\tau_\sigma =
0$ then $\sigma \in K_0 = K_{\tau_\sigma}$.  Now if $0 < \tau_\sigma
\leq 1$, condition (e) gives
$$
\bigcap_{p \in \sigma} D_p(\tau_\sigma) ~~=~~
\bigcap_{p \in \sigma} \; \bigcap_{t \in [0,\tau_\sigma)} D_p(t) ~~=~~
\bigcap_{t \in [0,\tau_\sigma)} \; \bigcap_{p \in \sigma} D_p(t) ~~=~~
\bigcap_{n \in \N^*} \; \bigcap_{p \in \sigma} D_p(\tau_\sigma - \tau_\sigma/n).
$$ Since the intersection of a sequence of decreasing non-empty
compact sets is non-empty, the right-hand side above is non-empty and
so is the left-hand side. Hence $\sigma \in K_{\tau_\sigma}$ and since
$\sigma \in K_0 \setminus K_1$, one has $0 \leq \tau_\sigma <1$. In
other words, the simplex $\sigma$ belongs to the complex till its
death time and disappears from the complex right after. For $t \in
[0,1)$, let $\Delta_t$ be the set of simplices with death time $t$.

\paragraph{Generic case.}
We first establish the lemma under the following generic condition:
  \begin{enumerate}
  \item[$(\star)$] For all $s \in [0,1)$, the set of simplices
    $\Delta_s$ is either empty or has a unique inclusion-minimal element.
  \end{enumerate}
At the end of the proof, we will explain how to get rid of this
genericity assumption. Consider $t \in [0,1)$ and suppose $\Delta_t
  \neq \emptyset$. We prove that the deletion of simplices $\Delta_t$
  from $K_t$ is a collapse. Let $\sigma_{\min}$ be the unique
  inclusion-minimal element of $\Delta_t$. Assuming we are in the
  generic situation, we do not need anymore conditions (c) and (d) but
  can replace them with the weaker conditions (c') and (d') obtained
  by considering constant maps for $\phi$ and $\phi'$.  Since
  $\bigcap_{p \in \sigma_{\min}} D_p(t) \neq \emptyset$ and
  $\bigcap_{p \in \sigma_{\min}} D_p(t + \eta) = \emptyset$ for all $0
  < \eta \leq 1 - t$, condition (d') implies that $\bigcap_{p \in
    \sigma_{\min}} D_p(t) = \{a\}$ for some $a \in A$. It is easy to
  see that $\Delta_t$ has a unique inclusion-maximal element
  $\sigma_{\max} = \{ p \in P \mid a \in D_p(t)\}$. Hence, $\Delta_t$
  consists of all cofaces of $\sigma_{\min}$ and these cofaces are
  faces of $\sigma_{\max}$. To prove that removing $\Delta_t$ from
  $K_t$ is a collapse, it suffices to establish that $\sigma_{\min}
  \neq \sigma_{\max}$. We proceed in two steps:

\medskip\noindent{\underline{Step 1:}} Let us prove that $a$ lies on
the boundary of $D_p(t)$ for all $p \in \sigma_{\min}$. For this, we
start by proving that $a$ lies on the boundary of at least one $D_p(t)$
for some $p \in \sigma_{\min}$. Suppose for a contradiction that $a$
belongs to the interior of $D_p(t)$ for all $p \in
\sigma_{\min}$. This implies that, for all $p \in \sigma_{\min}$,
there exists an open neighborhood $U_p$ of $a$ such that $a \in U_p
\subset D_p(t)$ and $a \in U = \bigcap_{p\in \sigma_{\min}} U_p
\subset \bigcap_{p\in \sigma_{\min}} D_p(t) = \{a\}$.  It follows that
$U = \{a\}$ and therefore $a$ is an isolated point of $A$.  Since $A$
is assumed to be connected, it entails that $A= \{a\}$. We thus reach a
contradiction since we obtain a case we have excluded. Defining
$\sigma_0 = \{ p \in \sigma_{\min} \mid a \in \partial D_p(t)\}$, we
have just proved that $\sigma_0 \neq \emptyset$.

Let us now prove that $\sigma_0 = \sigma_{\min}$. Suppose for a
contradiction that $\sigma_0$ is a proper subset of
$\sigma_{\min}$. As before, we can define an open set $U$ such that $a
\in U \subset D_p(t)$ for all $p \in \sigma_{\min} \setminus
\sigma_0$. We have
$$
a \; \in \;  \bigcap_{p \in \sigma_0} D_p(t) \cap U \; \subset \; \bigcap_{p
  \in \sigma_{\min}} D_p(t) = \{a\}.
$$
Setting $X = \bigcap_{p \in \sigma_0} D_p(t)$, we thus have $X \cap U
= \{a\}$ which is open in the subspace topology on $X$. Since $A$ is
$T_1$ separable, the subset $X \setminus \{a\}$ is also open in the
subspace topology on $X$. It follows that $X = \{a\} \cup (X \setminus
\{a\})$ is the union of two disjoint relatively open subsets. Since
$X$ is connected by condition (c'), one of the two subsets must be
empty. The only possibility is that $X \setminus \{a\} = \emptyset$
and $X = \bigcap_{p \in \sigma_0} D_p(t) = \{a\}$.

Because of (a), for all $t<t' \leq 1$, we get that $\bigcap_{p \in
  \sigma_0} D_p(t') \subset \bigcap_{p \in \sigma_0} D_p(t)^\circ =
\{a\}^\circ= \emptyset$. It follows that the death time of $\sigma_0$
is $t$ and the minimality of $\sigma_{\min}$ implies that $\sigma_0 =
\sigma_{\min}$, yielding a contradiction. Thus, $a$ lies on the
boundary of $D_p(t)$ for all $p \in \sigma_{\min}$.

\medskip\noindent{\underline{Step 2:}} Let us prove that
$\sigma_{\max} \neq \sigma_{\min}$.  By condition (b), we have $a \in
A = \bigcup_{p \in P} D_p(1)$ and therefore $a$ belongs to $D_q(1)$
for some $q \in P$. Since $t < 1$, condition (a) implies that $a \in
D_q(1) \subset D_q(t)^\circ$ and therefore $q \in
\sigma_{\max}$. On the other hand, $a \not \in \partial D_q(t)$
and therefore $q \not \in \sigma_{\min}$. It follows that
$\sigma_{\max} \neq \sigma_{\min}$.

\paragraph{Getting rid of the genericity assumption.} If we are not in
the generic case, the idea is to apply a small perturbation to the
family $\Dset$ which will leave unchanged $K_0$ and $K_1$ and such
that after perturbation (1) $\Dset$ will still satisfy the hypotheses
of the lemma; (2) the generic condition ($\star$) will hold. We say
that two simplices $\sigma_1$ and $\sigma_2$ are {\em in conjunction
  at time $t$} if they are both inclusion-minimal elements of
$\Delta_t$ for some $t \in [a,b)$. We say that $t$ is an {\em event
    time} if $\Delta_t \neq \emptyset$.  Consider two simplices that
  are in conjunction at time $t$, say $\sigma_1$ and
  $\sigma_2$. Suppose $q \in \sigma_1$ and $q \not \in
  \sigma_2$. Consider an increasing continuous bijection $\psi
  : [0,1] \to [0,1]$ that differs from identity only in a small
  neighborhood of $t$ that does not include any other event
  times. Replacing $D_q(t)$ by $D_q \circ \psi(t)$ and leaving unchanged $D_p(t)$
  for all $p \in P \setminus \{q\}$, we change
  the time at which $\sigma_1$ disappears while keeping unchanged the
  time at which $\sigma_2$ disappears. After this operation,
  $\sigma_1$ and $\sigma_2$ are not in conjunction
  anymore. Furthermore, the operation does not create any new pair of
  simplices in conjunction. By repeating this operation a finite
  number of times, we thus get a new collection as required.
\end{proof}

\paragraph{Remark.}
Somewhat surprisingly, condition (c) of
Lemma~\ref{lemma:abstract-tool} is weaker than the condition required
by the Nerve Lemma for guaranteeing that the simplicial complex $K_t =
\Nerve{\{D_p(t) \mid p\in P \}}$ is homotopy equivalent to $A$ at some
particular value of $t \in [0,1]$.  In particular, if the Nerve Lemma
applies at time $t=0$, that is, if $\bigcap_{p \in \sigma} D_p(0)$ is
either empty or contractible for all $\emptyset \neq \sigma \subset P$
and if furthermore the five conditions of
Lemma~\ref{lemma:abstract-tool} hold, then $K_t$ will have the right
homotopy type for all $t\in [0,1]$.

\subsection{Towards the nerve of $\alpha$-robust coverings}
\label{section:robust-coverings}

To state and prove our third theorem, we need some definitions. Given
a subset $X \subset \R^d$, we call the intersection of all balls of
radius $\alpha$ containing $X$ the {\em $\alpha$-hull of $X$} and
denote it by $\Hull{X}{\alpha}$. By construction, $\Hull{X}{\alpha}$
is convex and $\Hull{X}{+\infty}$ is the convex hull of
  $X$. Setting $\Clenchers{X}{\alpha} = \{ z \in \R^d \mid X \subset
B(z,\alpha) \}$, we have
$$
\Hull{X}{\alpha} = \bigcap_{z \in \Clenchers{X}{\alpha}} B(z,\alpha).
$$ Notice that $\Clenchers{X}{\alpha}$ is also convex; see Figure
\ref{fig:cells-strictly-decreasing}, left. Indeed, if two balls
$B(z_1,\alpha)$ and $B(z_2,\alpha)$ contain $X$, then any ball
$B(\lambda_1 z_1 + \lambda_2 z_2, \alpha)$ with $\lambda_1 + \lambda_2
= 1$, $\lambda_1 \geq 0$ and $\lambda_2 \geq 0$ also contains
$X$. Furthermore, if $X$ is compact, so is $\Clenchers{X}{\alpha}$.

\begin{definition}[$\alpha$-robust coverings] \label{definition:robust-coverings}
  A covering $\Cset = \{ C_v \mid v \in V\}$ of $A$ is {\em
    $\alpha$-robust} if (1) each set in $\Cset$ can be enclosed in an
  open ball with radius $\alpha$; (2) $\Nerve{\Cset} = \Nerve{\{ A
    \cap \Hull{C_v}{\alpha} \mid v \in V\}}$.
\end{definition}

Of course, one may wonder if $\alpha$-robust coverings of a shape $A$
often arise in practice. Section \ref{section:nice-triangulations}
will address this issue. For now we focus on establishing properties
of $\alpha$-robust coverings.


\begin{lemma}
\label{lemma:NiceCoveringsRecoverHomotopyType}
  If $\Cset$ is a finite compact $\alpha$-robust covering of $A$ and
  $0 \leq \alpha < \Reach{A}$, then $\Nerve{\Cset} \simeq A$.
\end{lemma}

\begin{proof}
  We apply the Nerve Lemma to the collection $\{ A \cap
  \Hull{C_v}{\alpha} \mid v \in V\}$. Clearly, $A = \bigcup_{v \in V}
  (A \cap \Hull{C_v}{\alpha})$. By Lemma~\ref{lemma:topological}, for
  all $\emptyset \neq \sigma \subset V$, the intersection $A \cap
  \bigcap_{v \in \sigma} \Hull{C_v}{\alpha}$ is either empty or
  contractible.
\end{proof}

Combining the above lemma and Theorem
\ref{theorem:RestrictedCechCaptureHomotopyType} we thus get that
$\Nerve{\Cset} \simeq \RCech{A}{P,\alpha}$ for all finite compact
$\alpha$-robust coverings $\Cset = \{C_v \mid v \in V\}$ of $A$ with
$0 \leq \alpha < \Reach{A}$. Next theorem strengthens this result and
states mild conditions on $P$ and $V$ under which there exists a
sequence of collapses transforming $\RCech{A}{P,\alpha}$ into a
simplicial complex isomorphic to $\Nerve{\Cset}$.

\begin{theorem}\label{theorem:collapsing-towards-nice-triangulations}
  Let $A$ be a compact set of $\R^d$ and $\alpha$ a real number such
  that $0 \leq \alpha < \Reach{A}$. Let $\Cset = \{ C_v \mid v \in
  V\}$ be a compact $\alpha$-robust covering of $A$. Let $P$ be a finite point set
  and suppose there exists an injective map $f : V \to P$ such that
  $C_v \subset \OpenBall{f(v)}{\alpha}$ for all $v \in V$. Then, there
  exists a sequence of collapses from $\RCech{A}{P,\alpha}$ to
  $f(\Nerve{\Cset}) = \{ f(\sigma) \mid \sigma \in
  \Nerve{\Cset} \}$.
\end{theorem}

\begin{proof} 
We build a family of compact sets $\mathcal{D} = \{D_p(t) \mid (p,t)
\in P \times [0,1] \}$ in such a way that if we let $K(t) = \Nerve{ \{
  D_p(t) \mid p \in P\}}$, then $\RCech{A}{P,\alpha} = K(0)$ and
$f(\Nerve{\Cset}) = K(1)$. We then prove that this family meets the
hypotheses of Lemma~\ref{lemma:abstract-tool}, implying that
$\RCech{A}{P,\alpha}$ can be transformed into $f(\Nerve{\Cset})$ by a
sequence of collapses obtained by increasing continuously $t$ from 0
to 1. To define the family $\mathcal{D}$, let us first associate to
every point $p \in P$ the set
$$
\Split(p) = 
\begin{cases}
  \Clenchers{C_v}{\alpha} & \mbox{if } f^{-1}(p) = \{v\}, \\
  \{p^+,p^-\} & \mbox{if } f^{-1}(p) = \emptyset,
\end{cases}
$$ where $p^+$ and $p^-$ are two points which are symmetric with
respect to $p$ and chosen such that $B(p^+,\alpha) \cap B(p^-,\alpha)
= \emptyset$. We then set
$$
D_p(t) = A \cap \bigcap_{s \in \Split(p)} \Ball{ (1-t)p + ts }{\alpha}.
$$ 
Let us check that $\RCech{A}{P,\alpha} = K(0)$ and
$f(\Nerve{\Cset}) = K(1)$. We claim that $\Split(p) \neq \emptyset$ for all
$p \in P$. Let us consider two cases. First, if $f^{-1}(p) =
\emptyset$, then by definition $\Split(p) = \{p^+,p^-\} \neq
\emptyset$. Second, if $f^{-1}(p) = \{v\}$, then $\Split(p)$ contains at
least $p$ since $C_v \subset \OpenBall{p}{\alpha}$ by
hypothesis. Thus, $D_p(0) = A \cap \bigcap_{s \in \Split(p)}
B(p,\alpha) = A \cap B(p,\alpha)$ and $K(0) =
\RCech{A}{P,\alpha}$. On the other hand, we have $D_p(1) = A \cap\bigcap_{s
  \in \Split(p)} B(s,\alpha)$ which we can rewrite as
$$
D_p(1) =
  \begin{cases}
    A \cap\Hull{C_v}{\alpha} & \mbox{if } f^{-1}(p) = \{v\}, \\
    \emptyset & \mbox{if } f^{-1}(p) = \emptyset.
  \end{cases}
  $$ 
Thus, $K(1) = f(\Nerve{\Cset})$.  Let us make some more
remarks. Writing $Z(p,t) = \{ (1-t)p + ts \mid s \in \Split(p)\}$,
we can express $D_p(t)$ as $A \cap \CapBalls{Z(p,t)}{\alpha}$. Since
$\Split(p)$ is compact, so is $Z(p,t)$ and by
Lemma~\ref{lemma:topological}, $D_p(t)$ is either empty or
contractible. Furthermore, $C_v \subset D_p(t)$ for all $p \in f(V)$,
showing that the collection of cells $D_p(t)$ cover the
shape. Applying the Nerve Lemma, we thus get that $K(t) \simeq A$ for
all $t \in [0,1]$. We are now ready to prove a stronger result, namely
that as $t$ increases continuously from 0 to 1, the only changes that
may occur in $K(t)$ are collapses. For this, it suffices to establish
that the family $\mathcal{D}$ defined above satisfies conditions (a),
(b), (c), (d) and (e) of Lemma~\ref{lemma:abstract-tool}.

\begin{figure}[htb]
    \centering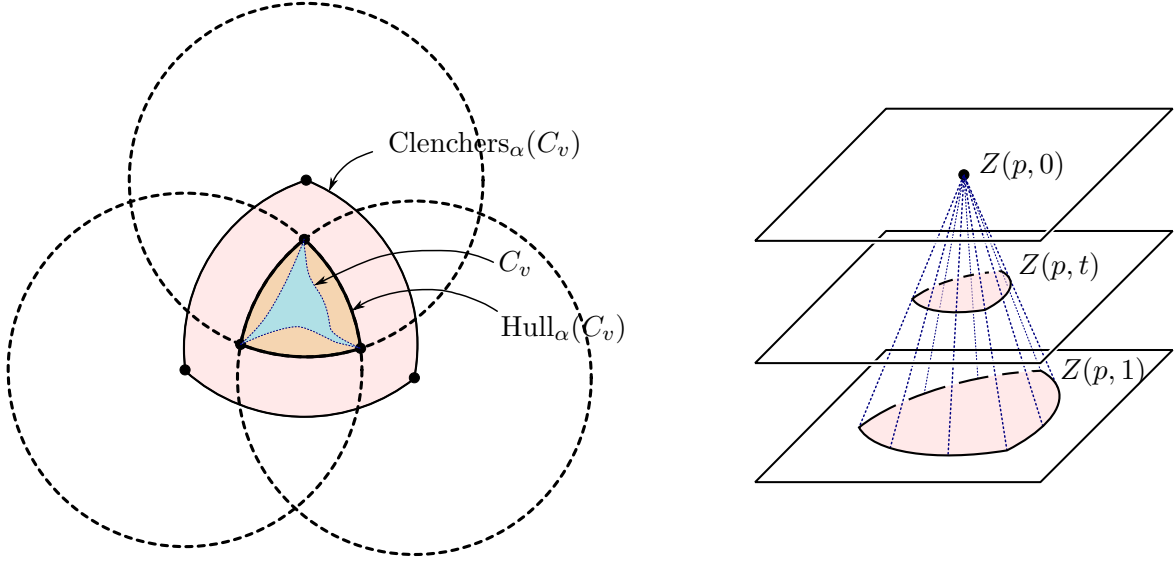
  \caption{Left: $\alpha$-hull and $\alpha$-clenchers of a planar cell
    $C_v$. Right: $Z(p,t)$ is the image of $Z(p,1)$ by an homothety
    centered at $p$ with scale factor $t$.
\label{fig:cells-strictly-decreasing}}
\end{figure}

\medskip \noindent (a) \hspace{1mm} Let us prove that for all $0 \leq
t < t' \leq 1$ and all $p \in P$, we have $D_p(t') \subset
D_p(t)^\circ$. If $f^{-1}(p) = \emptyset$, this is easy to
see. Suppose $f^{-1}(p) = \{v\}$; see Figure
\ref{fig:cells-strictly-decreasing}. We note that $Z(p,1) = \Split(p)
= \Clenchers{C_v}{\alpha}$ is convex and by construction, so are all
$Z(p,t)$ for all $t \in [0,1]$.  Since $C_v \subset
\OpenBall{p}{\alpha}$, it follows that $p$ belongs to the interior of
$Z(p,1)$ and $Z(p,t) \subset Z^\circ(p,t')$ for all $0 \leq t < t'
\leq 1$. This implies that $D_p(t') \subset D_p(t)$.  It remains to
show that no point of $D_p(t')$ is in $\partial D_p(t)$.

Suppose for a contradiction that $d \in \partial D_p(t) \cap D_p(t')$
and let $r = \max\{\|d-z\| \mid z \in Z(p, t)\}$. The real number $r$
is well-defined since $Z(p, t)$ is compact. From $d \in \partial
D_p(t)$ we can easily deduce that $r \geq \alpha$. Indeed, otherwise
some neighborhood of $d$ would belong to $D_p(t)$ which is impossible.
Let $z \in Z(p, t)$ be such that $\|z-d\| = r$.  Since $Z(p, t)
\subseteq Z^\circ(p, t')$, there is $z' \in Z(p, t')$ such that
$\|z'-d\| > r \geq \alpha$. But, this contradicts $d
\in D_p(t')$.

\medskip \noindent (b) \hspace{1mm} Clearly, $\bigcup_{p \in P} D_p(1) = A$.

\medskip \noindent (c) \hspace{1mm} Given $\sigma \subset P$ and a map
$\phi: \sigma \to [0,1]$, we introduce the set
$$
\Dset(\sigma,\phi) = \bigcap_{p \in \sigma} D_p \circ \phi(p) = A \cap\bigcap_{p \in
  \sigma} \bigcap_{z \in Z(p)} B((1-\phi(p))p+\phi(p)z,\alpha).
$$ By Lemma~\ref{lemma:topological}, the intersection
$\Dset(\sigma,\phi)$ is
either empty or connected.

\medskip \noindent (d) \hspace{1mm} Consider $\sigma \subset P$ and a
map $\phi: \sigma \to [0,1)$ such that $\Dset(\sigma,\phi) \neq
  \emptyset$. Let us prove that if $\Dset(\sigma,\phi') = \emptyset$
  for all maps $\phi': \sigma \to [0,1]$ with $\phi' \succ \phi$, then
  $\Dset(\sigma,\phi)$ is a singleton. Assume, by contradiction, that
  $\Dset(\sigma,\phi)$ contains two points $x_1$ and $x_2$ and let us
  prove that we can find a map $\phi': \sigma \to [0,1]$ such that
  $\phi' \succ \phi$ and $\Dset(\sigma,\phi') \neq \emptyset$. Take
  $\alpha'$ such that $\alpha < \alpha' < \Reach{A}$. Since $A \cap
  \Hull{\{x_1,x_2\}}{\alpha'}$ contains both $x_1$ and $x_2$, it is
  non-empty and therefore contractible by
  Lemma~\ref{lemma:topological}; see Figure \ref{fig:clenchers}.  In
  particular, there is a path connecting the points $x_1$ and $x_2$ in
  $A \cap \Hull{\{x_1,x_2\}}{\alpha'}$. This path has to intersect the
  largest ball $B$ contained in $\Hull{\{x_1,x_2\}}{\alpha'}$ and
  therefore $A \cap B \neq \emptyset$.  For $\xi >0$ sufficiently
  small we have
 $$
A \cap \Offset{B}{\xi}  ~~ \subset ~~ A \cap \Hull{\{x_1,x_2\}}{\alpha}
~~ \subset ~~ \Dset(\sigma,\phi).
$$ By moving slightly the centers of the balls defining
$\Dset(\sigma,\phi)$, that is, by replacing the map $\phi$ by
a map $\phi' \succ \phi$ such that $\phi'(p) - \phi(p)$ is small enough for all $p
\in \sigma$, we get a new set $\Dset(\sigma,\phi')$ that
still contains $B$.  Since $\emptyset \neq A \cap B$, we thus get that
$\Dset(\sigma,\phi') \neq \emptyset$, reaching a
contradiction.

\begin{figure}[htb]
    \centering\small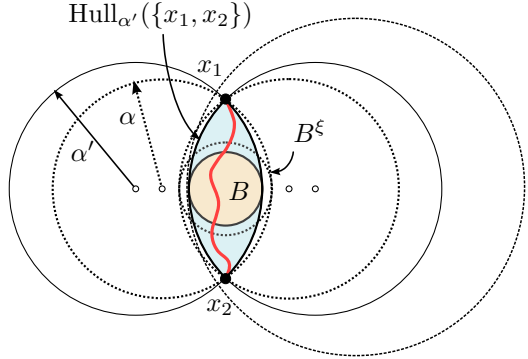
  \caption{Notation for the proof of Theorem
    \ref{theorem:collapsing-towards-nice-triangulations}.
\label{fig:clenchers}}
\end{figure}

\medskip \noindent (e) \hspace{1mm} It is not difficult to see that
$\bigcap_{t \in [0,\tau)} Z(p,t) = Z(p,\tau)$ and $D_p(\tau) =
  \bigcap_{t \in [0,\tau)} D_p(t)$ for $0 < \tau \leq 1$.
\end{proof}


\section{Nicely triangulable spaces}
\label{section:nice-triangulations}


Given a space $A$ and a finite sample $P$ of $A$, we are seeking a
sequence of collapses that transform the \v Cech complex of $P$ with
scale parameter $\alpha$ into a {\em triangulation} of $A$. We recall
that a {\em triangulation} of $A$ is a simplicial complex whose
underlying space is homeomorphic to $A$. If $A$ has a triangulation,
then $A$ is said to be {\em triangulable}. In particular, we know that
compact smooth manifolds are triangulable
\cite{whitney2005geometric}. Unfortunately, the proof involves
barycentric subdivisions whose dual meshes are not likely to have
convex cells and therefore have little chance of being $\alpha$-robust
coverings. And yet, we know that if a triangulation $T$ of a space $A$
is the nerve of an $\alpha$-robust covering of $A$, then the previous
section provides conditions under which $\RCech{A}{P,\alpha}$ can be
transformed into $T$ by a sequence of collapses.  This raises the
question of whether, given a space $A$ and a scale parameter $\alpha$,
it is possible to find a triangulation $T$ of $A$ which is the nerve
of some $\alpha$-robust covering of $A$. In this section, we focus on
the question and present examples of spaces enjoying this property.

\begin{figure}[htb]
  \def\svgwidth{0.84\linewidth}
    \centering\small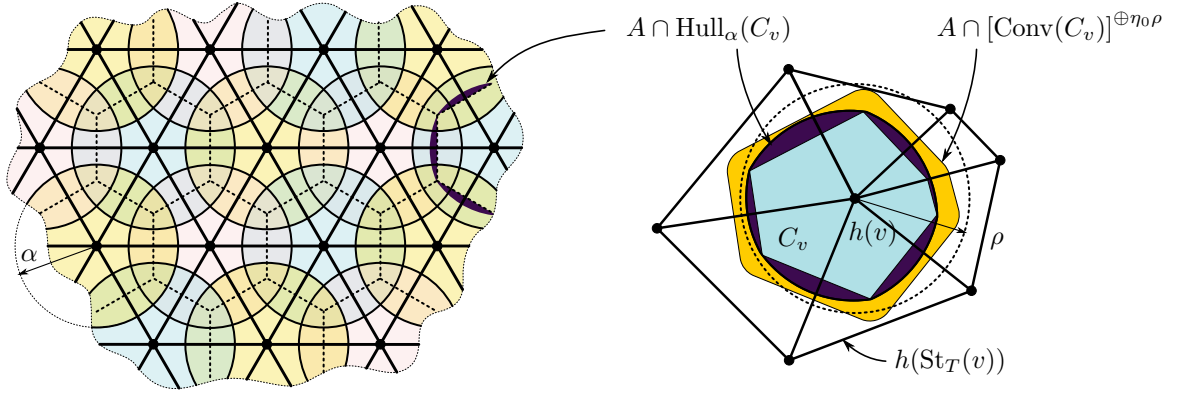
  \caption{Left: the collection of disks and the collection of Voronoi
    regions both form an $\alpha$-robust covering of the plane. Right:
    A triangulation is nice in our context when, among other things,
    it is the nerve of a collection of cells $C_v$ with size $\rho$
    such that $A \cap \left[\Conv{C_v}\right]^{\oplus \eta_0 \rho}
    \subset h(\Star{v}{T})$ for some $\eta_0>0$. This property will be
    preserved by $C^{1,1}$ diffeomorphisms for $\rho$ small
    enough. \label{fig:robust-coverings}}
\end{figure}

As a warm-up, we study the easy case $A = \R^2$; see Figure
\ref{fig:robust-coverings}. Consider a Delaunay triangulation $T$ of
$\R^2$ with vertex set $V$ and write $C_v = \{ x \in \R^2 \mid \|x-v\|
\leq \|x-u\| \mbox{ for all }u \in V \}$ for the Voronoi cell of $v
\in V$. Setting $\Cset = \{ C_v \mid v \in V\}$ for the collection of
Voronoi cells, we have that $T = \Nerve{\Cset}$. If all angles in $T$
are acute, then the Voronoi cell $C_v$ is contained in the star of $v$
and so is $\Hull{C_v}{\alpha}$ for $\alpha$ large enough. In
particular, by choosing carefully $V$ and $\alpha$, we can ensure that
$\Cset$ is an $\alpha$-robust covering of the plane.

To facilitate our discussion for more general spaces $A$, we first
introduce some more notations and definitions. Given an abstract
simplicial complex $T$, we let $g : V \to \R^{|V|-1}$ be an injective
map that sends the vertex set $V$ of $T$ to affinely independent
points of $\R^{|V|-1}$. The {\em underlying space} of $T$ is the point
set $|T| = \bigcup_{\sigma \in T} |\sigma|$, where $|\sigma|$ is the
geometric simplex obtained by taking the convex hull of
$g(\sigma)$. If $v$ is a vertex of $T$, the {\em open star} of $v$ in
$T$, denoted by $\Star{v}{T}$, is the union of the relative interiors
of $|\sigma|$ for all $\sigma$ of $T$ that contain $v$
\cite{munkres1993elements}. By definition, the set $\Star{v}{T}$ is
thus an open subset of $|T|$. For brevity, we shall write $h(v)$
instead of $h(|v|)$ and $h(\sigma)$ instead of $h(|\sigma|)$.  Writing
$\Conv{X} = \Hull{X}{+\infty}$ for the convex hull of $X \subset
\R^d$, we introduce the following definition:

\begin{definition}[nice triangulation]
  \label{definition:nice-triangulation}
  Let $\rho$ and $\delta$ be two positive real numbers.  A
  triangulation $T$ of $A \subset \R^d$ is said to be {\em
    $(\rho,\delta)$-nice} with respect to $(h,\Cset)$ if $h$ is a
  homeomorphism from $|T|$ to $A$, $\Cset = \{ C_v \mid v \in V\}$ is
  a finite compact covering of $A$ such that $\Nerve{\Cset} = T$ and
  the following conditions hold:
  \begin{enumerate}[{\em (i)}]
  \item $h(\sigma) \subset \bigcup_{v \in \sigma} C_v$ for all
    simplices $\sigma \in T$;
  \item $C_v \subset \OpenBall{h(v)}{\rho}$ for all $v \in V$;
  \item $A \cap \left[\Conv{C_v}\right]^{\oplus \delta} \subset h(\Star{v}{T})$ for all $v \in V$.
  \end{enumerate}
\end{definition}

The use of $\Conv{C_v}$ in the last item of the definition is
motivated by the following geometric lemma:

\begin{lemma}
\label{lemma:relation-hull-conv}
Let $X \subset \R^d$ be a non-empty compact set and $B(c,\rho)$ its
smallest enclosing ball. For all $\alpha$ and $\delta$ such that
$\alpha \geq \rho$ and $\alpha - \sqrt{\alpha^2 - \rho^2} \leq
\delta$, the following inclusion holds: $\Hull{X}{\alpha} \subset
\left[\Conv{X}\right]^{\oplus \delta}$.
\end{lemma}

\begin{figure}[htb]
  \def\svgwidth{0.45\linewidth}
  \centering\small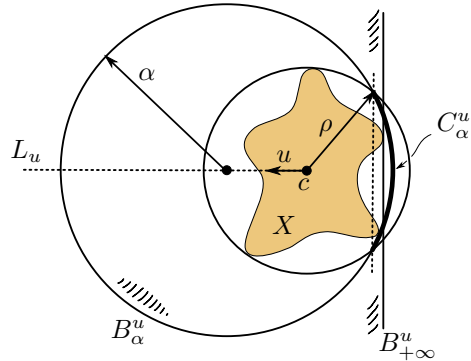
  \caption{Notation for the proof of
    Lemma~\ref{lemma:relation-hull-conv}.
    \label{fig:proof-technical-2}}
\end{figure}

\begin{proof}
  See Figure~\ref{fig:proof-technical-2}. Consider a unit
  vector $u \in \Sspace^{d-1}$ and let $L_u$ be the half-line
  emanating from $c$ in direction~$u$. Let $B^{u}_{\alpha}$ denote the
  ball with radius $\alpha$ centered on $L_u$ containing $X$ and whose
  center is furthest away from $c$. By construction, $X$ is contained
  in the intersection of the two balls $B^{u}_{\alpha}$ and
  $B(c,\rho)$. The boundary of $B^{u}_{\alpha} \cap B(c,\rho)$
  consists of two spherical caps and we let $C^{u}_{\alpha}$ be the
  one lying on the sphere bounding $B^{u}_{\alpha}$. Observe that $X$
  has a non-empty intersection with $C^{u}_{\alpha}$ and for all
  $\beta \geq \alpha$, the ball $B^{u}_{\beta}$ intersects
  $C^{u}_{\alpha}$. The largest distance between a point of
  $C^{u}_{\alpha}$ and $B^{u}_{\beta}$ is upper bounded by the height
  of $C^{u}_{\alpha}$ which is less than or equal to $\alpha -
  \sqrt{\alpha^2 - \rho^2} \leq \delta$. We thus get that
  $B^{u}_{\alpha} \subset
  \Offset{[B^{u}_{\beta}]}{\delta}$. Considering this inclusion over
  all directions $u$ for $\beta = +\infty$ yields the result.
\end{proof}

It follows that if $T$ is a $(\rho,\delta)$-nice triangulation of $A$
with respect to $(h,\Cset)$, we are able to derive conditions on
$\alpha$, $\rho$ and $\delta$ which guarantee that $\Cset$ is an
$\alpha$-robust covering of $A$.

\begin{lemma}
  \label{lemma:robust-covering}
    Let $A$ be a compact set of $\R^d$ and suppose $T$ is a
    $(\rho,\delta)$-nice triangulation of $A$ with respect to
    $(h,\Cset)$. Then $\Cset$ is an $\alpha$-robust covering of $A$
    whenever the following two conditions are fulfilled: (1) $\rho \leq
    \alpha$ and (2) $\alpha - \sqrt{\alpha^2 - \rho^2} \leq \delta$.
\end{lemma}

\begin{proof}
  Suppose $\Cset = \{ C_v \mid v \in V\}$ and let $v \in V$. By Lemma~\ref{lemma:relation-hull-conv}, $\Hull{C_v}{\alpha} \subset
  \left[\Conv{C_v}\right]^{\oplus \delta}$; see Figure
  \ref{fig:robust-coverings}, right. It follows that $C_v \subset A \cap
  \Hull{C_v}{\alpha} \subset h(\Star{v}{T})$ from which we deduce the
  sequence of inclusions 
  $$T = \Nerve{ \Cset } \subset \Nerve{ \{ A \cap
    \Hull{C_v}{\alpha} \mid v \in V \} } \subset \Nerve{ \{ h(
    \Star{v}{T} ) \mid v \in V \} } = T.
  $$ The nerves on the left and on
  the right are equal, showing that $\Nerve{ \Cset } = \Nerve{ \{ A
    \cap \Hull{C_v}{\alpha} \mid v \in V \} }$.
\end{proof}

Observe that if $T$ is a $(\rho,\eta_0\rho)$-nice triangulation of $A$
for some $\eta_0>0$, then conditions (1) and (2) of
Lemma~\ref{lemma:robust-covering} are satisfied for $\delta = \eta_0
\rho$ as soon as $\rho$ is small enough. Of course, the difficult
question is whether such a triangulation $T$ can always be found for
arbitrarily small $\rho$.

\begin{definition}[nicely triangulable]
  We say that $A \subset \R^d$ is nicely triangulable if we can find
  $\rho_0 > 0$ and $\eta_0>0$ such that for all $0 < \rho < \rho_0$,
  there is a $(\rho,\eta_0\rho)$-nice triangulation of $A$.
\end{definition}

\begin{theorem}
  Suppose $A \subset \R^d$ is nicely triangulable. For every $0 <
  \alpha < \Reach{A}$, there exists $\varepsilon_0>0$ such that for all
  finite point set $P \subset \R^d$ and all $0 < \varepsilon < \varepsilon_0$
  satisfying $A \subset \Offset{P}{\varepsilon}$, the complex
  $\RCech{A}{P,\alpha}$ can be transformed into a triangulation of $A$
  by a sequence of collapses.
\end{theorem}

\begin{proof}
  By definition, we can find $\rho_0 > 0$ and $\eta_0>0$ such that for
  all $0 < \rho < \rho_0$, there is a $(\rho,\eta_0\rho)$-nice
  triangulation $T$ of $A$ with respect to $(h,\Cset)$. Let us choose
  $\rho$ small enough so that $\rho < \alpha$ and $\alpha -
  \sqrt{\alpha^2 - \rho^2} \leq \eta_0 \rho$.  Lemma~\ref{lemma:robust-covering} then implies that $\Cset$ is a compact
  $\alpha$-robust covering of $A$. Set $e(T,h) = \frac{1}{2} \inf
  \|h(v_1) - h(v_2)\|$ where the infimum is over all pairs of vertices
  $v_1 \neq v_2$ of $T$ and let $\varepsilon_0$ be the minimum of
  $e(T,h)$ and $\alpha-\rho$. Consider a function $f : \Vertexset{T}
  \to P$ that maps each vertex $v$ to a point of $P$ closest to
  $h(v)$. Note that $f$ is injective, $\|h(v)-f(v)\| \leq \varepsilon$
  and $C_v \subset \OpenBall{f(v)}{\alpha}$ for all $v \in
  V$. Applying Theorem
  \ref{theorem:collapsing-towards-nice-triangulations} yields the
  existence of a sequence of collapses from $\RCech{A}{P,\alpha}$ to
  $f(T)$.
\end{proof}

The next theorem provides a few examples of nicely triangulable
manifolds.

\begin{theorem}\label{lemma:TemplateOfNicelyTriangulableManifold}
The following embedded manifolds are nicely triangulable:
\begin{enumerate}
\item The unit 2-sphere $\Sspace^2 = \{ x = (x_1,x_2,x_3) \in \R^3
  \mid \sum_{i=1}^3 x_i^2 = 1 \}$;
\item The flat torus $\Tspace^2 = \{ x=(x_1,x_2,x_3,x_4) \in \R^4 \mid x_1^2 + x_2^2
  =1 \text{ and } x_3^2+x_4^2 = 1 \}$;
\item The $m$-dimensional Euclidean space $\R^m$, embedded in $\R^d$ for some $m \leq d$.
\end{enumerate}
\end{theorem}

\begin{proof}
  For $A \in \{ \Sspace^2, \Tspace^2, \R^m\}$, we proceed as
  follows. We build a triangulation $T$ parameterized by some integer
  $n$ and consider a map $h : |T| \to A$. The integer $n$ will control
  the size of elements in $h(T)$: the larger $n$ the smaller the image
  of simplices under $h$.  We then consider the barycentric
  subdivision $K$ of $T$ and associate to each vertex $v$ of $T$ the
  cell $C_v = \bigcup_{\sigma \ni v} h(\sigma)$. The
  collection of cells $C_v$ forms a covering $\Cset$ of $A$.  In
    the three cases, it is not difficult to see that we can find
  $\eta_0>0$ such that $T$ is $(\rho, \eta_0 \rho)$-nice with respect
  to $(h,\Cset)$ for some $\rho>0$. Furthermore, the value of $\rho$
  can be made as small as desired by increasing $n$. We thus conclude
  that $A$ is nicely triangulable. Below, we just describe how $T$ and
  $h$ are chosen in each case.

\begin{figure}[htb]
  \def\svgwidth{0.90\linewidth}
  \centering\small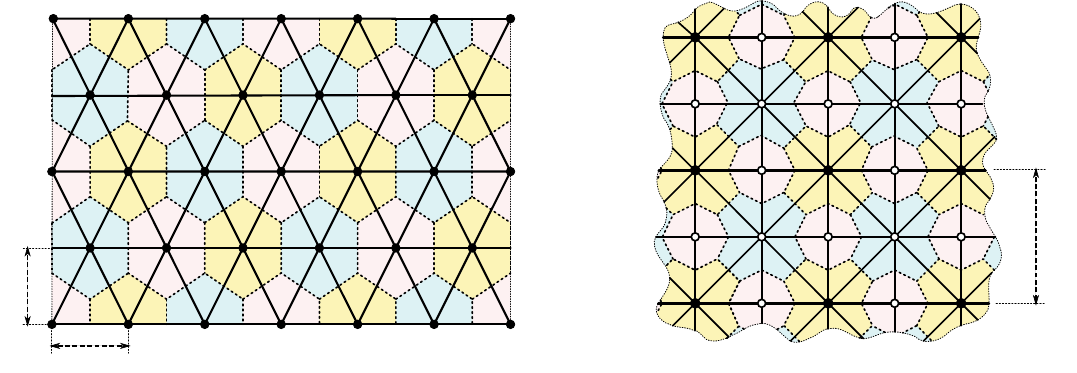
  \caption{ Triangulating $\Tspace^2$ (left) and $\R^2$ (right).
    \label{fig:tillings}
  }
\end{figure}

\medskip \noindent{\bf 1. $\Sspace^2$ is nicely triangulable.}  We
start with an icosahedron centered at the origin and subdivide each
triangular face into $4^n$ equilateral triangles. Notice that all
vertices of the resulting triangulation $T$ have degree $6$ but the
$12$ vertices in the original icosahedron which have degree $5$.  The
triangulation $T$ is then projected onto the sphere, using the
projection map $h : |T| \to \Sspace^2$ defined by $h(x) =
\frac{x}{\|x\|}$.

\medskip \noindent{\bf 2. $\Tspace^2$ is nicely triangulable.}  The
map $H: \R^2 \rightarrow \Tspace^2$ defined by $H(s,t) = (\cos s, \sin
s, \cos t, \sin t)$ is locally isometric and its restriction $h
:[0,2\pi)^2 \to \Tspace^2$ is an homeomorphism. The idea is to build a
  periodic tiling of $\R^2$ made up of identical isosceles triangles
  as in Figure \ref{fig:tillings}, left. Let $a$ and $b$ be the
  respective height and basis of the triangles. Consider two integers
  $n$ and $k$ such that $na = kb = 2 \pi$. Taking $k = \lfloor
  \frac{\sqrt{3}}{2} n \rfloor$ we get that the ratio $\frac{a}{b}$
  tends to $\frac{\sqrt{3}}{2}$ as $n\rightarrow +\infty$. Thus, the
  map $H$ turns the periodic tiling of $\R^2$ into a triangulation of
  $\Tspace^2$ whose triangles become arbitrarily close to equilateral
  triangles with edge length $b$ as $n \to +\infty$.

\medskip \noindent{\bf 3. $\R^m$ is nicely triangulable.}  We start
with a cubical regular grid and define $T$ as the barycentric
subdivision of that grid; see Figure \ref{fig:tillings},
right. Precisely, for each cell in the grid, we insert one vertex at
its centroid. So each edge is subdivided into 2 edges sharing the
inserted vertex. We then recursively subdivide the cells by ascending
dimension. Each cubical $k$-cell has $2k$ cubical $(k-1)$-cells on its
boundary. We subdivide each $k$-cell as a cone whose apex is the
inserted vertex and whose basis is the subdivided boundary of that
cell. We claim that all stars in $T$ are convex. Indeed each vertex in
$T$ is the centroid of an initial cubical cell of dimension between
$0$ and $m$. Consider the vertex $v$ that was inserted at the center
of the $k$-dimensional cubical cell $D_v$ and let us describe the set
of vertices $V_v$ in the link of $v$ in $T$. The vertices of $V_v$ can
be partitioned in two subsets. The first subset contains vertices in
the $k$-flat that supports $D_v$ while the second subset contains
vertices in the $(d-k)$-flat passing through $v$ and orthogonal to the
$k$-flat supporting $D_v$. The vertices in the first subset lie on the
boundary of $D_v$ and the vertices in the second subset lie on the
boundary of a $(m-k)$-cube. Since in both flats of respective
dimension $k$ and $m-k$ the vertices in $V_v$ are in convex position,
it results that vertices in $V_v$ are in convex position in $\R^m$. As
a result, it can be proved (details are skipped) that the star of $v$
is the convex hull of $V_v$. Finally, we let $h$ be the identity map
and $n$ the inverse of the size of the grid.
\end{proof}


We now establish that the property of being nicely triangulable is
preserved by $C^{1,1}$ diffeomorphisms between manifolds. Let us make
precise what we mean in Theorem
\ref{theorem:NicelyTriangulabePreservedByDiffeo} by embedded $C^{1,1}$
$k$-manifolds.  A $C^{1,1}$ function is a differentiable function with
a Lipschitz derivative. A $C^{1,1}$ structure on a manifold is an
equivalence class of atlases whose transition functions are
$C^{1,1}$. Finally, $C^{1,1}$ diffeormorphisms between $C^{1,1}$
manifolds are defined accordingly. In Theorem
\ref{theorem:NicelyTriangulabePreservedByDiffeo}, we restrict our
attention to shapes which are $C^{1,1}$ compact manifolds without
boundary embedded in $\R^d$ and whose embeddings are themselves
regular and $C^{1,1}$ (for the differential structure induced by
$\R^d$) , where ``regular'' means that the derivative of the embedding
has full rank everywhere. We will say that such shapes are compact
$C^{1,1}$ manifolds embedded in $\R^d$ for short. The assumption of
regular embeddings entails the existence of well-defined tangent
affine spaces. A compact manifold embedded in $\R^d$ is $C^{1,1}$ if
and only if it has a positive reach \cite{federer-59}.



\begin{theorem}\label{theorem:NicelyTriangulabePreservedByDiffeo}
Let $M$ and $M'$ be two compact $C^{1,1}$ $k$-manifolds without
  boundary embedded respectively in $\R^d$ and $\R^{d'}$ and $\Phi :
M \rightarrow M'$ a $C^{1,1}$ diffeomorphism. $M$ is nicely
triangulable if and only if $M'$ is nicely triangulable.
\end{theorem}

The proof is given in the Appendix.


\section{Discussion}
\label{section:conclusion}


The paper leaves unanswered a few questions that we discuss now:
\begin{enumerate}[(1)] \denselist
\item Our result assumes the shape to be nicely triangulated. In the
  paper, we list a few simple spaces which enjoy this property. Is it
  possible to extend the list to a larger class of spaces? We
  conjecture that compact smooth $k$-manifolds embedded in $\R^d$ are
  nicely triangulable. Indeed, for $k=2$, it is known that any compact
  connected surface (without boundary) embedded in $\R^3$ is
  homeomorphic to either the 2-sphere or a connected sum of $g$ tori
  for $g \geq 1$. Hence, thanks to Theorem
  \ref{theorem:NicelyTriangulabePreservedByDiffeo}, it would suffice to
  provide a template of nicely triangulable surface of genus $g$ for
  each $g \geq 2$, in a way similar to what we did for $g=0$ and
  $g=1$. Unfortunately, for higher dimensional manifolds, one cannot
  rely anymore on an existing classification. Another approach has to
  be considered.

\item Our proof is not constructive. Indeed, the order in which to
  collapse faces in the \v Cech complex is determined by sweeping
  space with a $t$-offset of the shape for decreasing values of
  $t$. Since the common setting consists in describing the shape
  through a finite sample, the knowledge of the $t$-offsets of the
  shape is lost. Nonetheless, is it possible to turn our proof into an
  algorithm? Can we do the same for Rips complexes?  A positive answer
  is even more desirable for the second class of complexes due to
  their computational tractability. We leave those questions open for
  future work.
\end{enumerate}


\bibliographystyle{abbrv}


\clearpage
\appendix


\section{$C^{1,1}$ diffeomorphisms preserve nicely triangulable manifolds}
\label{section:missing-proofs}


The goal of this section is to prove Theorem
\ref{theorem:NicelyTriangulabePreservedByDiffeo}.  

\begin{proof}[Proof of Theorem \ref{theorem:NicelyTriangulabePreservedByDiffeo}]
Let $x \in M$. Since $M$ is a compact $C^{1,1}$ $k$-manifold embedded
in $\R^d$, there exists a $k$-dimensional affine space $T_M(x) \subset
\R^d$ tangent to $M$ at $x$. Let $\pi_x : M \to T_M(x)$ be the
orthogonal projection onto the tangent space $T_M(x)$ and let $\pi_{\Phi(x)}'
: M' \to T_{M'}(\Phi(x))$ the orthogonal projection onto
$T_{M'}(\Phi(x))$. Since $M$ and $M'$ are compact, we can find two
constants $K$ and $K'$ independent of $x$ such that:
\begin{align}
  \label{eq:K}
  \forall y \in M, \quad &\|y-\pi_x(y)\| < K \|y-x\|^2 \\
  \label{eq:K'}
  \forall y \in M', \quad &\|y-\pi_{\Phi(x)}'(y)\| < K' \|y-\Phi(x)\|^2
\end{align}
Given $t_0 > 0$, we consider the open set $U_x = M
\cap \OpenBall{x}{t_0}$ and adjust $t_0$ in such a way that
\begin{enumerate}
\item The restriction $\pi_x : U_x \to \pi_x(U_x)$ is an homeomorphism
  for all $x \in M$;
\item The restriction $\pi_{\Phi(x)}' : \Phi(U_x) \to \pi_{\Phi(x)}'(\Phi(U_x))$ is also an
  homeomorphism for all $x \in M$.
\end{enumerate}

For sake of conciseness, we only sketch a justification for the
  existence of such a $t_0>0$.  The local property (i.e. the existence
  of $t_0>0$ for a given $x \in M$) follows easily from the definition
  of embedded $C^{1,1}$ $k$-manifolds. Indeed, the assumption of a
  regular embedding entails that $\pi_x$ has full rank derivative at
  $x$ and the inverse function theorem can be applied to get the local
  property.  In order to get a uniform $t_0>0$ (the requested global
  property) one can establish first the following strengthening of the
  local property: For any $x \in M$, there is $t_x>0$ such that for
  any $y \in M \cap B^\circ (x, t_x)$, the restriction of $\pi_y$ to
  $M \cap B^\circ (x, t_x)$ is a $C^1$ homeomorphism.  Compactness of
  $M$ can then be used in the usual manner to get a uniform $t_0$.


The collection of pairs $\{(U_x,\pi_x)\}_{x \in M}$ forms an atlas in
the $C^{1,1}$ structure of $M$. Similarly, the collection of pairs
$\{(\Phi(U_x),\pi_{\Phi(x)}')\}_{x \in M}$ forms an atlas in the $C^{1,1}$
structure of $M'$. Let $\mathcal{T}_M(x)$ be the linear space
associated to $T_M(x)$ and denote by $D\Phi_x$ the derivative of
$\Phi$ at $x$, seen as a linear map between $\mathcal{T}_M(x)$ and
$\mathcal{T}_{M'}(\Phi(x))$. Since $M$ is compact, there is a constant 
$K_\Phi$ independent of $x$ such that:
\begin{equation}
\label{eq:K-Phi}
\forall y \in M, \quad \|\Phi(y)-\Phi(x)-D\Phi_x(\pi_x(y)-x)| < K_\Phi \|y-x\|^2,
\end{equation}
and two positive numbers $\kappa_2 \geq \kappa_1 > 0$, again
independent of $x$ by compactness of $M$, such that
\begin{equation}
\label{eq:derivative-bounds}
\forall u \in
\mathcal{T}_M(x), \quad  \kappa_1 \|u\| \leq \| D\Phi_x (u) \| \leq
\kappa_2 \|u\|.  
\end{equation}
Consider the affine function $\hat{\Phi}_x : T_M(x) \rightarrow
T_{M'}(\Phi(x))$ defined by $\hat{\Phi}_x (y) = \Phi(x) + D\Phi_x
(y-x)$. Combining Equations (\ref{eq:K}) (\ref{eq:K'}) and (\ref{eq:K-Phi}), we can
find a constant $L_\Phi$ independent of $x$ such that for all $t <
t_0$ and all compact sets $A \subset M \cap B(x,t)$ :
\begin{equation}
\label{eq:L-Phi}
  d_H(\hat{\Phi}_x \circ \pi_x(A) , \pi_{\Phi(x)}' \circ \Phi(A) ) < L_\Phi t^2
\end{equation}

Now, assume that $M$ is nicely triangulable and let us prove that $M'$
is also nicely triangulable. By definition, we can find $\rho_0 > 0$
and $\eta_0>0$ such that, for all $0 < \rho < \rho_0$, there is a
$(\rho,\eta_0\rho)$-nice triangulation $T$ of $M$ with respect to some
$(h, \Cset)$. Suppose $\Cset = \{C_v \mid v \in V\}$ and consider the
covering $\Cset' = \{ \Phi(C_v) \mid v \in V\}$, the homeomorphism $h'
= \Phi \circ h : |T| \to M'$, the real numbers $\rho'=2\kappa_2\rho$
and $\eta'_0 = \frac{\kappa_1 \eta_0 - 5 L_\Phi \rho}{2\kappa_2}$. Let
us prove that by choosing $\rho$ small enough, $T$ is a
$(\rho',\eta_0'\rho')$-nice triangulation of $M'$ with respect to
$(\Cset',h')$. In other words, we need to check that conditions (ii)
and (iii) of Definition \ref{definition:nice-triangulation} are
satisfied for $\Cset = \Cset'$, $h=h'$, $\rho=\rho'$ and $\delta =
\eta_0'\rho'$. Take $v \in V$ and set $x=h(v)$, $C = C_v$, $S =
h(\Star{v}{T})$.

\medskip (ii) By definition of $T$, we have $x \in C \subset
\OpenBall{x}{\rho}$. Taking the image of this relation under $\Phi$
and choosing $\rho>0$ small enough, we get that $\Phi(x) \in \Phi(C)
\subset \Phi(\OpenBall{x}{\rho}) \subset
\OpenBall{\Phi(x)}{\rho'}$. The last inclusion is obtained by
combining Equations (\ref{eq:K}), (\ref{eq:K-Phi}) and
(\ref{eq:derivative-bounds}).

\medskip (iii) Let us choose a positive real number $\rho < \min{ \{
  \rho_0, \frac{t_0}{2}\}}$ small enough to ensure that $\eta_0' > 0$
and let us prove that $M' \cap \left[ \Conv{C} \right]^{\oplus
  \eta_0'\rho'} \subset S$. By choice of $T$ as a
$(\rho,\eta_0\rho)$-nice triangulation of $M$ with respect to
$(h,\Cset)$, we have that $M \cap \Conv{C}^{\oplus \eta_0 \rho}
\subset S$. Furthermore, $C \subset B(x,\rho)$ and $S \subset
B(x,2\rho)$. Thus, by choosing $\rho < \frac{t_0}{2}$, we have $S
\subset U_x$ and
\begin{equation*}
  U_x \cap  \Conv{C}^{\oplus \eta_0 \rho} \subset S.
\end{equation*}
Taking the image by the homeomorphism $\pi_x : U_x \to \pi_x(U_x)$ on
both sides and using $\pi_x(A \cap B) = \pi_x(A) \cap \pi_x(B)$ we get
\begin{equation*}
  \pi_x(\Conv{C}^{\oplus \eta_0 \rho}) \subset \pi_x(S).
\end{equation*}
Let $B_k(0,r)$ denote the $k$-dimensional ball of
$\mathcal{T}_M(x)$  centered at the origin with radius $r$. Writing
$A\oplus B = \{ a + b \mid a \in A, b \in B\}$ for the Minkowski sum
of $A$ and $B$, it is not too difficult to prove that $\pi_x(A^{\oplus
  \delta}) = \pi_x(A) \oplus B_k(0,\delta)$. It follows that
\begin{equation*}
  \pi_x(\Conv{C}) \oplus B_k(0,\eta_0\rho) \subset \pi_x(S).
\end{equation*}
Taking the image under $\hat{\Phi}_x$ on both
sides we get
\begin{equation*}
  \hat{\Phi}_x \circ \pi_x(\Conv{C})  \oplus D \Phi_x B_k(0,\eta_0\rho)
  \subset \hat{\Phi}_x \circ \pi_x(S)).
\end{equation*}
Let $B'_k(0,r)$ denote the $k$-dimensional ball of
$\mathcal{T}_{M'}(\Phi(x))$  centered at the origin with radius $r$.
Using Equation (\ref{eq:derivative-bounds}) we get that
$B_k'(0,\kappa_1\eta_0\rho) \subset D \Phi_x B_k(0,\eta_0\rho)$. Since
$\hat{\Phi}_x$ and $\pi_x$ are both affine, so is the composition and therefore
$\hat{\Phi}_x \circ \pi_x(\Conv{C}) = \Conv{\hat{\Phi}_x \circ
  \pi_x(C)}$. It follows that
\begin{equation*}
\Conv{\hat{\Phi}_x \circ \pi_x(C)} \oplus B_k'(0,\kappa_1 \eta_0 \rho)
\subset \hat{\Phi}_x  \circ \pi_x(S).
\end{equation*}
Recalling that $C \subset B(x,\rho)$ and $S \subset B(x,2\rho)$ and
combining the above inclusion with Equation~(\ref{eq:L-Phi}) we obtain
\begin{equation*}
  \Conv{\pi'_x \circ \Phi(C)} \oplus B_k(0, \kappa_1 \eta_0 \rho - 5 L_\Phi \rho^2)
  \subset \pi'_x \circ \Phi(S).
\end{equation*}
Interchanging $\operatorname{Conv}$ and $\pi_x'$, noting that $\eta'_0\rho' =
\kappa_1 \eta_0 \rho - 5 L_\Phi \rho^2$ and using
$\pi'_x(A^{\oplus \delta}) = \pi'_x(A) \oplus B_k(0,\delta)$ we get
\begin{equation*}
  \pi'_x(\Conv{\Phi(C)}^{ \oplus \eta_0' \rho'})
  \subset \pi'_x \circ \Phi(S).
\end{equation*}
Since $\pi'_x : \Phi(U_x) \to \pi_{\Phi(x)}'(\Phi(U_x))$ is
homeomorphic, we thus obtain $M' \cap \Conv{\Phi(C)}^{\oplus \eta_0'
  \rho} \subset \Phi(S)$ as desired.
\end{proof}


\end{document}